\documentclass[a4paper,12pt]{article}

\usepackage{setspace} 
\usepackage{amsmath}
\usepackage{amssymb}
\usepackage{ascmac}
\usepackage{amsthm}
\usepackage{amscd}
\usepackage{color}
\usepackage{enumerate}
\usepackage[dvips]{graphicx}
\usepackage{float}
\usepackage{wrapfig}

  \makeatletter
    
    \@addtoreset{equation}{section}
  \makeatother

\theoremstyle{definition}
\newtheorem{definition}{Definition}[section]
\theoremstyle{plain}
\newtheorem{theorem}[definition]{Theorem}

\newtheorem{lemma}[definition]{Lemma}
\newtheorem{proposition}[definition]{Proposition}
\theoremstyle{remark}
\newtheorem{remark}[definition]{Remark}
\newcommand{\Do}{\partial\!\!\!/}

\begin{document}

\title{ Quaternifications and Extensions of  Current Algebras on \(S^3\) }
\author{Tosiaki Kori and Yuto Imai 
\\Department of Mathematics\\
Graduate School 
of Science and Engineering\\
Waseda University,\\Tokyo 169-8555, Japan
\\email{ kori@waseda.jp, y.imai@aoni.waseda.jp}}
\date{}
\maketitle

	\begin{abstract}
Let \(\mathbf{H}\) be the quaternion algebra.   Let \(\mathfrak{g}\) be a complex Lie algebra and let \(U(\mathfrak{g})\) be the enveloping algebra of \(\mathfrak{g}\).   The quaternification \(\mathfrak{g}^{\mathbf{H}}=\,(\,\mathbf{H}\otimes U(\mathfrak{g}),\,[\quad,\quad]_{\mathfrak{g}^{\mathbf{H}}}\,)\) of \(\mathfrak{g}\) is defined by the bracket 
\begin{equation}
\bigl[\,\mathbf{z}\otimes X\,,\,\mathbf{w}\otimes Y\,\bigr]_{\mathfrak{g}^{\mathbf{H}}}\,=\,(\mathbf{z}\cdot \mathbf{w})\otimes\,(XY)\,-\,
(\mathbf{w}\cdot\mathbf{z})\otimes (YX)\,,\nonumber
\end{equation}
for  \(\mathbf{z},\,\mathbf{w}\in \mathbf{H}\) and  \(X,\,Y\in U(\mathfrak{g})\,\).    Let \(S^3\mathbf{H}\) be the ( non-commutative ) algebra of \(\mathbf{H}\)-valued smooth mappings over \(S^3\) and let  \(S^3\mathfrak{g}^{\mathbf{H}}=S^3\mathbf{H}\otimes U(\mathfrak{g})\).    
The Lie algebra structure on \(S^3\mathfrak{g}^{\mathbf{H}}\) is induced naturally from that of \(\mathfrak{g}^{\mathbf{H}}\).     We introduce a 2-cocycle on \(S^3\mathfrak{g}^{\mathbf{H}}\)  by the aid of a tangential vector field on \(S^3\subset \mathbf{C}^2\) and 
 have the corresponding central extension \(S^3\mathfrak{g}^{\mathbf{H}} \oplus(\mathbf{C}a)\).    
  As a subalgebra of \(S^3\mathbf{H}\) we have the algebra of Laurent  polynomial spinors \(\mathbf{C}[\phi^{\pm}]\) spanned by a complete orthogonal system of eigen spinors \(\{\phi^{\pm(m,l,k)}\}_{m,l,k}\) of the tangential Dirac operator on \(S^3\).   Then  \(\mathbf{C}[\phi^{\pm}]\otimes U(\mathfrak{g})\) is a Lie subalgebra of \(S^3\mathfrak{g}^{\mathbf{H}}\).   We have the central extension    
 \(\widehat{\mathfrak{g}}(a)=
(\,\mathbf{C}[\phi^{\pm}]
\otimes U(\mathfrak{g}) \,) \oplus(\mathbf{C}a)\) as a Lie-subalgebra of \(S^3\mathfrak{g}^{\mathbf{H}} \oplus(\mathbf{C}a)\).     
Finally we  have a Lie algebra \(\widehat{\mathfrak{g}}\) which is obtained by adding to \(\widehat{\mathfrak{g}}(a)\) a derivation \(d\) which acts on  \(\widehat{\mathfrak{g}}(a)\) by the Euler  vector field \(d_0\).    That is the \(\mathbf{C}\)-vector space \(\widehat{\mathfrak{g}}=\left(\mathbf{C}[\phi^{\pm}]\otimes U(\mathfrak{g})\right)\oplus(\mathbf{C}a)\oplus (\mathbf{C}d)\) endowed with the bracket
\begin{eqnarray*}
&& \bigl[\,\phi_1\otimes X_1+ \lambda_1 a + \mu_1d\,,\phi_2\otimes X_2 + \lambda_2 a + \mu_2d\,\,\bigr]_{\widehat{\mathfrak{g}}}
\, =	
(\phi_1\phi_2)\otimes (X_1\,X_2)
\, -\,(\phi_2\phi_1)\otimes (X_2X_1)\\[0.2cm]
\qquad&&+\mu_1d_0\phi_2\otimes X_2-\mu_2d_0\phi_1\otimes X_1 
 +  (X_1\vert X_2)c(\phi_1,\phi_2)a\,.
\end{eqnarray*}
When  \(\mathfrak{g}\) is a simple Lie algebra with its Cartan subalgebra \(\mathfrak{h}\)  we shall investigate the weight space decomposition of \(\widehat{\mathfrak{g}}\) with respect to the  subalgebra \(\widehat{\mathfrak{h}}=
(\phi^{+(0,0,1)}\otimes \mathfrak{h} )\oplus(\mathbf{C}a)
\oplus(\mathbf{C}d)\).   
\end{abstract}
\maketitle


\medskip

\section{Introduction}
The set of smooth mappings from a manifold to a Lie algebra has  been a subject of investigation both from a purely mathematical standpoint and from quantum field theory.   In quantum field theory they appear as a current algebra or an infinitesimal gauge transformation group.    Loop algebras are the simplest example.   Loop algebras and their representation theory have been fully worked out.   A loop algebra valued in  a  simple Lie algebra or its complexification turned out to behave like a simple Lie algebra  and the highly developed theory of finite dimensional Lie algebra was extended to such loop algebras.    Loop algebras appear in the simplified model of quantum field theory where the space is one-dimensional and many important facts in the representation theory of loop algebra  were first discovered by physicists.    As is well known 
A. Belavin et al. \cite{BPZ} constructed two-dimensional conformal field theory based on the irreducible representations of Virasoro algebra.    
It turned out that in many applications to field theory one must deal with certain extensions of  the associated  loop algebra rather than the loop algebra itself.    The central extension of a loop algebra is called  an affine Lie algebra and the highest weight theory of finite dimensional Lie algebra was extended to this case.   \cite{K}, \cite{K-W},  \cite{P-S} and \cite{W} are good references to study these subjects. 

In this paper we shall investigate a  generalization of affine Lie algebras to the Lie algebra of mappings from three-sphere \(S^3\)  to a Lie algebra.     As an affine Lie algebra is a central extension of the Lie algebra of smooth mappings from \(S^1\) to the complexification of a Lie algebra, so our objective is an extension of the Lie algebra of smooth mappings from \(S^3\) to the quaternification of a Lie algebra.        As for the higher dimensional generalization of loop groups, J. Mickelsson introduced an abelian exension  of current groups  \(Map(S^3,SU(N))\) for \(N\geq 3\), \cite{M}.   It is related to the Chern-Simons function on the space of \(SU(N)\)-connections and the associated current algebra \( Map(S^3,su(N))\) has an abelian extension \( Map(S^3,su(N))\oplus {\cal A}_3^{\ast}\) by the affine dual of the space \({\cal A}_3\) of connections over \(S^3\), \cite{Ko4}.     
In \cite{P-S} it was shown that, for any smooth manifold \(M\) and a simple Lie algebra \(\mathfrak{g}\), there is a  universal central extension of the Lie algebra \(Map(M,\mathfrak{g})\).     The kernel of the extension is given by the space of 
 complex valued 1-forms modulo exact 1-forms; \(\Omega^1(M)/d\Omega^0(M)\).     It implies that any extension is a weighted linear combination of extensions obtained as a pull back of the universel extension of the loop algebra \(L\mathfrak{g}\) by a smooth loop \(f:S^1\longrightarrow M\).   While two-dimensional conformal field theory is based on this central extension,  we would like to provide a mathematical tool that could help constructing a four-dimensional conformal field theory.
This is why we are dealing with central extensions of the Lie algebra of smooth mappings from \(S^3\) to the quaternification of a Lie algebra.    Now we shall give a brief explanation of each section.

Let  \( \mathbf{H}\)  be the quaternion numbers.     In this paper we shall denote a quaternion \(a+jb\in \mathbf{H}\) by \(\left(\begin{array}{c}a\\b\end{array}\right)\).   This comes from the identification of \( \mathbf{H}\)  with the matrix algebra
 \[
\mathfrak{mj}(2,\mathbf{C})=\left\{\left(\begin{array}{cc}
a&-\overline b\\[0.2cm] b&\overline a\end{array}\right)\,:\quad a,b\in\mathbf{C}\right\}.\]
\(\mathbf{H}\) becomes an associative algebra and 
the Lie algebra structure \( (\, \mathbf{H}, [\,,\,]_{ \mathbf{H}}\,)\) is induced on it.    The trace of  \(\mathbf{a}=\left(\begin{array}{c}a\\b\end{array}\right)\in \mathbf{H}\) is defined by \(tr\,\mathbf{a}=a+\overline a\).    For \(\mathbf{u}, \mathbf{v}, \mathbf{w}\in \mathbf{H}\) we have 
\(tr\,( [\mathbf{u},\mathbf{v}]_{ \mathbf{H}}\cdot\mathbf{w}\,)=\,tr\,( \mathbf{u}\cdot [\mathbf{v}, \mathbf{w}]_{ \mathbf{H}}\,)\).

Let \((\, \mathfrak{g}\, ,\, \bigl[\quad ,\quad\bigr]_{\mathfrak{g}}\,) \) be a complex Lie algebra.     Let 
\(U(\mathfrak{g})\) be the enveloping algebra.    The quaternification  of \(\mathfrak{g}\) is defined as the vector space  \(\mathfrak{g}^{\mathbf{H}}= \mathbf{H}\otimes U(\mathfrak{g})\) endowed with the bracket 
\begin{equation}
\bigl[\,\mathbf{z}\otimes X\,,\,\mathbf{w}\otimes Y\,\bigr]_{\mathfrak{g}^{\mathbf{H}}}\,=\,(\mathbf{z}\cdot \mathbf{w})\otimes\,(XY)\,-\,(\mathbf{w}\cdot\mathbf{z})\otimes (YX)\,,\label{bracket} 
\end{equation}
for  \(\mathbf{z},\,\mathbf{w}\in \mathbf{H}\) and  \(X,\,Y\in U(\mathfrak{g})\,\).  
It extends the Lie algebra structure \((\mathfrak{g},\,\bigl[\quad,\quad\bigr]_{\mathfrak{g}}\,) \)  to  \(\left(\mathfrak{g}^{\mathbf{H}},\, \bigl[\quad\,,\quad\,\bigr]_{\mathfrak{g}^{\mathbf{H}}}\right)\) .     The quaternions \(\mathbf{H}\) give also a half spinor representation of \(Spin(4)\).   That is, \(\Delta=\mathbf{H}\otimes \mathbf{C}=\mathbf{H}\oplus \mathbf{H}\) gives an irreducible complex representation of the Clifford algebra \({\rm Clif}( \mathbf{R}^4)\): 
\({\rm Clif } ( \mathbf{R} ^4)\otimes  \mathbf{C}\,\simeq\, {\rm End} (\Delta )\), and 
\(\Delta\) decomposes into irreducible representations \(\Delta^{\pm}=\mathbf{H}\) of \({\rm Spin}(4)\).   
   Let \(S^{\pm}=\mathbf{C}^2\times \Delta^{\pm}\) be the trivial even ( respectively odd ) spinor bundle.   A section of spinor bundle is called a spinor.     The space of even half spinors \(C^{\infty}(S^3,S^+)\) is identified with the space \(S^3\mathbf{H}=Map(S^3,\mathbf{H})\).   Now the space    \(S^3{\mathfrak{g}}^{\mathbf{H}}\,=S^3\mathbf{H}\otimes U(\mathfrak{g})\,\) becomes a Lie algebra with respect to  the  bracket: 
 \begin{equation}
 [\,\phi \otimes X\, , \,\psi  \otimes Y\,]_{ S^3{\mathfrak{g}}^{\mathbf{H}}} =  ( \phi \psi ) \otimes \,(XY) \,
- \,(\psi \phi)\,  \otimes (YX)
 , 
\end{equation}
for 
  \(X,Y\in U(\mathfrak{g})\,\) and \(\,\phi,\,\psi\,\in S^3\mathbf{H}\, \).    In the sequel we shall abbreviate the Lie bracket \([\,,\,]_{ S^3{\mathfrak{g}}^{\mathbf{H}}}\)  simply to \([\,,\,]\).     Such an abbreviation will be often adopted for other Lie algebras. 
  
    Recall that the central extension of a loop algebra \(L\mathfrak{g}=\mathbf{C}[z,z^{-1}]\otimes\mathfrak{g}\,\) is the Lie algebra \((L\mathfrak{g}\oplus \mathbf{C}a\,,[\,,\,]_c \,)\)   given by  the bracket
\[[P\otimes X,Q\otimes Y]_c=PQ\otimes [X,Y]+(X\vert Y)c(P,Q)a\,,\]
with the aid of the 2-cocycle \(c(P,Q)=\frac{1}{2\pi}\int_{S^1}(\frac{d}{dz}P)\,Q\,dz\), where \((\cdot \vert \cdot)\)  is a non-degenerate invariant symmetric bilinear form on \(\mathfrak{g}\), \cite{K}.   
We shall give an analogous 2-cocycle on  \(S^3\mathbf{H}\,\).    
Let \(\theta\) be the vector field on \(S^3\) defined by 
\begin{equation}
\theta = z_1 \frac{\partial}{\partial z_1} + z_2 \frac{\partial}{\partial z_2}
- \Bar {z_1} \frac{\partial}{\partial \Bar{z_1}} - \Bar{z_2} \frac{\partial}{\partial \Bar{z_2}}.
\end{equation}
For \(\varphi=\left(\begin{array}{c}u\\ v\end{array}\right)\in S^3\mathbf {H}\), we put 
\[
\Theta\,\varphi=\frac{1}{2\sqrt{-1}}\,\left(\begin{array}{c}\theta\, u\\[0.2cm] \theta\, v\end{array}\right).\]
Let \(c:\,S^3{\mathbf{H}}\times S^3{\mathbf{H}}\longrightarrow \mathbf{C}\) be the bilinear form given by 
\begin{equation}
c(\phi_1,\phi_2)\,=\,
\,\frac{1}{2\pi^2}\int_{S^3}\,tr   [\,\Theta\,\phi_1\,\cdot \phi_2\,] d\sigma\,,\quad \phi_1,\,\phi_2 \in S^3\mathbf{H}.
  \end{equation}
 \(c\) defines a 2-cocycle on the algebra \(S^3\mathbf{H}\,\).    That is, \(c\,\) satisfies the following equations:
\[
c(\phi_1,\,\phi_2)=\,- \,c(\phi_2,\,\phi_1)\,\]
and 
\[c(\,\phi_1\cdot\,\phi_2\,,\,\phi_3)+c(\,\phi_2\cdot\,\phi_3\,,\,\phi_1)+c(\,\phi_3\cdot\,\phi_1\,\,,\,\phi_2)=0\,.
\]
We extend \(c\) to the 2-cocycle on \(S^3{\mathfrak{g}}^{\mathbf{H}}\) by 
\begin{equation}
c(\,\phi_1\otimes X\,,\,\phi_2\otimes Y\,)=\,(X \vert Y)\,c(\phi_1,\phi_2),
\end{equation}
where \((\,\cdot\vert\,\cdot)\) is the non-degenerate invariant symmetric bilinear form on \(\mathfrak{g}\) extended to \(U(\mathfrak{g})\).

Let \(a\) be an indefinite element.  
The Lie algebra extension of \(S^3{\mathfrak{g}}^{\mathbf{H}}\,\) by the 2-cocycle \(c\) is the 
\(\mathbf{C}\)-vector space \(S^3{\mathfrak{g}}^{\mathbf{H}}\oplus \mathbf{C}a\,\) endowed with the following bracket:   
 \begin{eqnarray}\label{liebra}
 [\,\phi \otimes X\, , \,\psi  \otimes Y\,]^{\wedge}&=&  ( \phi \cdot\psi ) \otimes \,(X\,Y) \,
- \,(\psi  \cdot\phi) \otimes (Y\,X)\, + \, c(\phi,\psi)(X\vert Y)\, a\,,\nonumber\\[0.2cm]
[a\,,\phi\otimes X\,]^{\wedge}&=&0\,,
\end{eqnarray}
for 
\(X,Y\in U(\mathfrak{g})\) and 
 \(\phi,\,\psi\,\in S^3{\mathbf{H} }\).
 
In section 2 we shall review the theory of spinor analysis after \cite{Ko2, Ko3}. 
 Let \(D : \,S^+\longrightarrow S^-\) be the ( half spinor ) Dirac operator.  
  Let \(D=\gamma_+(\frac{\partial}{\partial n}-\Do)\) be the polar decomposition on \(S^3\subset\mathbf{C}^2\) of the Dirac operator, where \(\Do\) is the tangential Dirac operator on \(S^3\) and \(\gamma_+\) is the Clifford multiplication of the unit normal derivative on \(S^3\).     The eigenvalues of \(\Do\) are given by \(\{\frac{m}{2},\,\,-\frac{m+3}{2}\,;\,m=0,1,\cdots \} \), with multiplicity \((m+1)(m+2)\).    We have an explicitly written formula for eigenspinors \(\left\{ \phi^{+(m,l,k)},\,\phi^{-(m,l,k)}\right\}_{0\leq l\leq m,\,0\leq k\leq m+1}\)  corresponding to the eigenvalue \(\frac{m}{2}\) and  \(-\frac{m+3}{2}\) respectively and they give rise to a complete orthogonal system in \(L^2(S^3, S^+)\).     A spinor \(\phi\) on a domain \(G\subset \mathbf{C}^2\) is called a {\it harmonic spinor} on \(G\) if \(D\phi=0\).
Each \(\phi^{+(m,l,k)}\) is extended to a harmonic spinor on \(\mathbf{C}^2\), while each \(\phi^{-(m,l,k)}\) is extended to a harmonic spinor on \(\mathbf{C}^2\setminus \{0\}\).   Every harmonic spinor \(\varphi\) on \(\mathbf{C}^2\setminus \{0\}\) has a Laurent series expansion by the basis \(\phi^{\pm(m,l,k)}\):
 \begin{equation}
 \varphi(z)=\sum_{m,l,k}\,C_{+(m,l,k)} \phi^{+(m,l,k)}(z)+\sum_{m,l,k}\,C_{-(m,l,k)}\phi^{-(m,l,k)}(z).
 \end{equation}
 If only finitely many coefficients are non-zero it is called a {\it spinor of Laurent polynomial type }.     The algebra of spinors of Laurent polynomial type is denoted by   \(\mathbf{C}[\phi ^{\pm}] \).    
  \(\mathbf{C}[\phi ^{\pm}] \) is a subalgebra of \(S^3\mathbf{H}\) that is algebraically generated by 
  \(\phi^{+(0,0,1)}=\left(\begin{array}{c}1\\0\end{array}\right)\), 
   \(\phi^{+(0,0,0)}=\left(
  \begin{array}{c}0\\ -1\end{array}\right) \),
     \( \phi^{+(1,0,1)} =\left(\begin{array}{c}z_2\\-\overline z_1\end{array}\right) \) and 
   \(\phi^{-(0,0,0)}=\left(\begin{array}{c}z_2\\\overline z_1\end{array}\right)\,\).     
   
 As a Lie subalgebra of \(S^3\mathfrak{g}^{\mathbf{H}}\),  \(\,\mathbf{C}[\phi^{\pm}]\otimes U(\mathfrak{g})\,\) has the central extension by the 2-cocycle \(c\,\).   That is,      
 the \(\mathbf{C}\)-vector space \(\widehat{\mathfrak{g}}(a)=\mathbf{C}[\phi ^{\pm}]\otimes U(\mathfrak{g})\oplus \mathbf{C}a\) endowed with the  Lie bracket (\ref{liebra})  becomes an extension of  \(\mathbf{C}[\phi^{\pm}]\otimes U(\mathfrak{g})\) with 1-dimensional center \(\mathbf{C}a\).   
  Finally we shall construct the Lie algebra which is obtained by adding to \(\widehat{\mathfrak{g}}(a)\)  a derivation \(d\) which acts on \(\widehat{\mathfrak{g}}(a)\) by the Euler vector field \(d_0\) on \(S^3\).     
   The Euler vector field is by definition
\(d_0\,= \frac{1}{2 }( z_1 \frac{\partial}{\partial z_1} + z_2 \frac{\partial}{\partial z_2}  +  \overline z_1 \frac{\partial}{\partial \overline z_1} +\overline z_2 \frac{\partial}{\partial \overline z_2}) \).    
 We have the following fundamental property of the cocycle \(c\,\).
 \[
c(\,d_0\phi_1\,,\phi_2\,)+
c(\,\phi_1\,, d_0\phi_2\,)=0.
\]
Let 
\(\widehat{\mathfrak{g}}=(\,\mathbf{C}[\phi^{\pm}]\otimes U(\mathfrak{g})\,)\oplus(\mathbf{C}a)\oplus (\mathbf{C}d)\).   We endow \(\widehat{\mathfrak{g}}\) with the bracket defined by 
 \begin{eqnarray*}
 [\,\phi \otimes X\, , \,\psi  \otimes Y\,]_{\widehat{\mathfrak{g}}} &= &
  [\,\phi \otimes X\, , \,\psi  \otimes Y\,]^{\wedge}
 \,,  \qquad 
 [\,a\,, \phi\otimes X\,] _{\widehat{\mathfrak{g}}}=0\,, \\[0.2cm]
 \, [\,d,\,a\,]_{\widehat{\mathfrak{g}}}&=& 0\,,\qquad
  [\,d, \phi \otimes X\,] _{\widehat{\mathfrak{g}}}=\,d_0 \phi \otimes X\, .
  \end{eqnarray*}
Then \((\,\widehat{\mathfrak{g}}\,,\,[\,\,,\,\,]_{\widehat{\mathfrak{g}}}\,)\) is an extension of the Lie algebra \(\widehat{\mathfrak{g}}(a)\) on which \(d\) acts as  \(d_0\).      
In section 4, when  \(\mathfrak{g}\) is a simple Lie algebra with its Cartan subalgebra \(\mathfrak{h}\,\),  we shall investigate  the weight space decomposition of \(\,\widehat{\mathfrak{g}}\) with respect to the  subalgebra \(\widehat{\mathfrak{h}}=(\phi^{+(0,0,1)}\otimes \mathfrak{h} )\oplus(\mathbf{C}a) \oplus(\mathbf{C}d)\), the latter is a commutative subalgebra and \(ad(\widehat{\mathfrak{h}})\)  acts on  \(\widehat{\mathfrak{g}}\) diagonally.    For this purpose we look at the representation of the adjoint action of \(\mathfrak{h}\) on the enveloping algebra \(U(\mathfrak{g})\).   Let \(\mathfrak{g}=\sum_{\alpha\in\Delta}\,\mathfrak{g}_{\alpha}
\) be the root space decomposition of \(\mathfrak{g}\).      Let \(\Pi=\{\alpha_i;\,i=1,\cdots,r={\rm rank}\,\mathfrak{g}\}\subset \mathfrak{h}^{\ast}\) be the set of simple roots and  \(\{\alpha_i^{\vee}\,;\,i=1,\cdots,r\,\}\subset \mathfrak{h}\) be the set of simple coroots.   The Cartan matrix \(A=(\,a_{ij}\,)_{i,j=1,\cdots,r}\) is given by \(a_{ij}=\left\langle \alpha_i^{\vee},\,\alpha_j \right\rangle\).      Fix a standard set of generators \(H_i=\alpha_i^{\vee}\), \(X_i= X_{\alpha_i}\in \mathfrak{g}_{\alpha_i}\),  \(Y_i= X_{-\alpha_i}\in \mathfrak{g}_{-\alpha_i}\), so that \([X_i,\,Y_j ]=H_j\delta_{ij}\), \([H_i,\,X_j]=-a_{ji}X_j\) and \([H_i,\,Y_j]=a_{ji}Y_j\). 
   We see that the set of weights of the representation \((\,U(\mathfrak{g}),\,ad(\mathfrak{h} ))\)  becomes
\begin{equation}
\Sigma=\{\,\sum_{i=1}^r\,k_i\alpha_i\in \mathfrak{h}^{\ast}\,;\quad k_i \in \mathbf{Z},\,i=1,\cdots,r\,\}
\end{equation}
The weight space of  \(\lambda\in\Sigma  \) is by definition 
\begin{equation}
\mathfrak{g}^U_{\lambda}\,=\,\{\xi\in U(\mathfrak{g})\,;\, ad(h)\xi=\lambda(h)\xi,\,\forall h \in\mathfrak{h}\},
\end{equation}
when \(\mathfrak{g}^U_{\lambda}\neq 0\).   
Then, given \(\lambda=\sum_{i=1}^r\,k_i\alpha_i\) , we have 
\begin{equation*}
\mathfrak{g}^U_{\lambda}=\mathbf{C}[Y_{1}^{q_1}\cdots\,Y_{r}^{q_r}\,H_1^{l_1}\cdots H_r^{l_r}\,X_{1}^{p_1}\cdots X_{r}^{p_r}\,
;\, p_i, \,q_i,\, l_i\in \mathbf{N}\cup 0,\, \,k_i=p_i-q_i\,, \,i=1,\cdots,r \, ]\,.
\end{equation*} 
 The weight space decomposition becomes 
\begin{equation}
U(\mathfrak{g})=\bigoplus_{\lambda\in \Sigma}\,\mathfrak{g}^U_{\lambda}\,,\quad \mathfrak{g}^U_0\supset U(\mathfrak{h}).
\end{equation}
Now we proceed to the representation \((\,\widehat{\mathfrak{g}},\,ad(\widehat{\mathfrak{h}})\,)\).   
The dual space $\mathfrak{h}^* $ of $\mathfrak{h} $ can be regarded naturally as a subspace of  $\,\widehat{\mathfrak{h}}^{\,\ast}$.   So  $\Sigma \subset \mathfrak{h}^*$ is seen to be a subset of $\,\widehat{\mathfrak{h}}^{\,*}$.    
 We define $\delta \in \widehat{\mathfrak{h}}^{\,*}$ by putting 
\(\left\langle\delta , \,h_i  \,\right\rangle =\,\left\langle\delta , a\right\rangle = 0\),  \(1 \leqq i \leqq  r\),  and 
 \(\left\langle\delta , d\right\rangle = 1\).   
 Then the set of weights $\widehat{\Sigma}$ of the representation $(\,\widehat{\mathfrak{g}} , \,ad(\widehat{\mathfrak{h}})\,)$ is
\begin{eqnarray}
\widehat{\Sigma} &=& \left\{ \frac{m}{2} \delta+  \lambda;\quad \lambda \in \Sigma\, ,\,m\in\mathbf{Z}\,\right\} \nonumber\\[0.2cm]
&& \bigcup \left\{ \frac{m}{2} \delta ;\quad  m\in \mathbf{Z}\, \right\}.
\end{eqnarray}
The weight space decomposition of \(\widehat{\mathfrak{g}}\) is given by
\begin{equation}
\widehat{ \mathfrak{g}}\,=\, \bigoplus_{m\in \mathbf{Z}}\, \widehat{ \mathfrak{g}}_{\frac{m}{2}\delta}\,\bigoplus\,\left(\,\bigoplus_{\lambda\in \Sigma,\,\,m\in \mathbf{Z}}\, 
\widehat{ \mathfrak{g}}_{\frac{m}{2}\delta+\lambda}\,\right)
\end{equation}
Each weight space is given as follows.\\
 
\begin{eqnarray*}  
\widehat{ \mathfrak{g}}_{\frac{m}{2}\delta+ \lambda}\,&=&\mathbf{C}[\phi ^{\pm};\,m\,] \otimes \mathfrak{g} _{ \lambda}^U\,\qquad\mbox{ for \(m\neq 0\) and \(\lambda\neq 0\,\),} \\[0.2cm]
 \widehat{ \mathfrak{g}}_{\frac{m}{2}\delta}&=&  \,\mathbf{C}[\phi^{\pm};\,m\,]  \otimes \mathfrak{g}^U_0\,\, \qquad\mbox{for  \(m\neq  0 \)  },\\[0.2cm]
  \widehat{ \mathfrak{g}}_{0\delta}&= &(\,\mathbf{C}[\phi^{\pm};0\,]  \otimes \mathfrak{g}^U_0\,)\oplus(\mathbf{C}a)\oplus(\mathbf{C}d)\,\supset\,\widehat{\mathfrak{h}}\,, 
\end{eqnarray*}
where
 \[\mathbf{C}[\phi^{\pm}; m\,] \,=\,\left\{\varphi\in \mathbf{C}[\phi^{\pm}] ;\,|z|^m\varphi(\frac{z}{|z|})=\varphi(z)\,\right\}.\]

\section*{Acknowledgement}
A r\'esum\'e of these results is appeared in \cite{K-I}.     The present article is devoted to the  explanation of these results with detailed proof.     
The authors would like to express their thanks to Professors Yasushi Homma of Waseda University  for his valuable objections to the early version of this paper.

\section{Quaternification of a Lie algebra} 

\subsection{Quaternion algebra}

The quaternions \(\mathbf{H}\) are formed from the real numbers \(\mathbf{R}\) by adjoining three symbols \(\,i,\,j,\,k\,\) satisfying the identities:
\begin{eqnarray}\label{q}
i^2&=&j^2=k^2=-1\,,\nonumber\\[0.2cm]
ij&=&-ji=k,\quad jk=-kj=i, \quad ki=-ik=j\,.
\end{eqnarray}   A general quaternion is of the form \(\,x=x_1+x_2i+x_3j+x_4k\,\) with 
\(x_1,x_2,x_3,x_4\in \mathbf{R}\).   By taking \(x_3=x_4=0\) the complex numbers \(\mathbf{C}\) are contained in \(\mathbf{H}\) if we identify \(i\) as the usual complex number.    Every quaternion 
\(x\) has a unique expression   
\(x=z_1+jz_2\) with \(z_1,z_2\in\mathbf{C}\).    This identifies \(\mathbf{H}\) with \(\mathbf{C}^2\) as \(\mathbf{C}\)-vector spaces.   
The quaternion multiplication will be from the right \(x\longrightarrow xy\) where 
\(y=w_1+jw_2\) with \(w_1,\,w_2\in \mathbf{C}\):
\begin{equation}
xy=( z_1+jz_2\,)(  w_1+jw_2\,)=(z_1w_1-\overline z_2 w_2)+j(\overline z_1w_2+z_2w_1) .
\label{rmulti}
\end{equation}
The multiplication  of a  \(g=a+jb\in \mathbf{H}\) to \(\mathbf{H}\) from the left yields an endomorphism in \(\mathbf{H}\): \(\{x\longrightarrow gx\}\in End_{\mathbf{H}}(\mathbf{H})\).
If we look on it under the identification \(\mathbf{H}\simeq\mathbf{C}^2\) mentioned above we have the \(\mathbf{C}\)-linear map
\begin{equation}
\mathbf{C}^2\ni\left(\begin{array}{c}z_1\\z_2\end{array}\right)\,\longrightarrow
\,\left(\begin{array}{cc}
a&-\overline b\\[0.2cm] b&\overline a\end{array}\right)\left(\begin{array}{c}z_1\\z_2\end{array}\right)\,\in \mathbf{C}^2\,.\end{equation}
This establishes the \(\mathbf{R}\)- linear isomorphism
\begin{equation}
\mathbf{H}\,\ni\, a+jb\,\stackrel{\simeq}{\longrightarrow}\, \left(\begin{array}{cc}
a&-\overline b\\[0.2cm] b&\overline a\end{array}\right)\,\in \mathfrak{mj}(2,\mathbf{C}),
\end{equation}
where we defined 
\begin{equation}
\mathfrak{mj}(2,\mathbf{C})=\left\{\left(\begin{array}{cc}
a&-\overline b\\[0.2cm] b&\overline a\end{array}\right)\,:\quad a,b\in\mathbf{C}\right\}.
\end{equation}
The complex matrices corresponding to \(i,\,j,\,k\in\mathbf{H}\) are 
\begin{equation}\label{3basis}
e_3=
\left(\begin{array}{cc}
i&0\\[0.2cm] 0&-i\end{array}\right)\,,\, e_2=\left(\begin{array}{cc}
0&-1\\[0.2cm] 1&0\end{array}\right)\,,\, e_1=\left(\begin{array}{cc}
0&-i\\[0.2cm] -i&0\end{array}\right)\,. 
\end{equation}
These are the basis of the Lie algebra $\mathfrak{su}$(2).
Thus we have the identification of the following objects
\begin{equation}
\mathbf{H}\,\simeq \mathfrak{mj}(2,\mathbf{C})\simeq \mathbf{R}\oplus \mathfrak{su}(2).\label{trivext}
\end{equation}
The correspondence between the elements is given by 
\begin{equation}
a+jb\equiv \left(\begin{array}{c}a\\b\end{array}\right)\,\longleftrightarrow 
\left(\begin{array}{cc}
a&-\overline b\\[0.2cm] b&\overline a\end{array}\right)
\,
\longleftrightarrow \,
s + pe_1+ qe_2 + re_3\,,
\label{correspond}
\end{equation}
where \(a=s+ir,\,b=q+ip\).

   \(\mathbf{H}\) becomes an associative algebra with  the multiplication law defined by 
\begin{equation}
\left(\begin{array}{c}z_1\\z_2\end{array}\right)\,\cdot \,\left(\begin{array}{c}w_1\\w_2\end{array}\right)\,=\,\left(\begin{array}{c}z_1w_1-\overline z_2w_2\\ \overline z_1w_2+z_2w_1\end{array}\right)\,,\label{productlaw}
\end{equation}
which is the rewritten formula of (\ref{rmulti}) and 
the right-hand side is the first row of the matrix multiplication
\[\left(\begin{array}{cc}
z_1&-\overline z_2\\[0.2cm] z_2&\overline z_1\end{array}\right)\,
\left(\begin{array}{cc}
w_1&-\overline w_2\\[0.2cm] w_2&\overline w_1\end{array}\right)\,.\]
It implies the 
Lie bracket of two vectors in \(\mathbf{H}\), that becomes   
\begin{equation}\label{Liebr}
\left [\,\left(\begin{array}{c}z_1\\z_2\end{array}\right),\,\left(\begin{array}{c}w_1\\w_2\end{array}\right)\,\right ] \,
=\left(
\begin{array}{c}z_2\overline w_2-\overline z_2w_2\\
(w_1-\overline w_1)z_2-(z_1-\overline z_1)w_2\end{array}\right).
\end{equation}

These expressions are very convenient to develop the analysis on \(\mathbf{H}\), and give an interpretation on the quaternion analysis by the language of spinor analysis.

\begin{proposition}
Let \(\mathbf{z}=\left(\begin{array}{c}z_1\\z_2\end{array}\right),\, \mathbf{w}=\left(\begin{array}{c}w_1\\w_2\end{array}\right)\in\mathbf{H}\,\). Then the trace of \(\mathbf{z}\cdot\mathbf{w}\in \mathbf{H}\,\simeq \mathfrak{mj}(2,\mathbf{C})\)  is given by
\begin{equation}
tr\,(\mathbf{z}\cdot\mathbf{w})\,=\,2{\rm Re}(z_1w_1-\overline z_2 w_2),
\end{equation}
and we have, for \(\mathbf{z}_1,\, \mathbf{z}_2,\,\mathbf{z}_3\,\in \mathbf{H}\), 
\begin{equation}
tr\,\left(\,[\,\mathbf{z}_1,\,\mathbf{z}_2\,]\,\cdot\mathbf{z}_3\,\right)\,=\,tr\,\left(\,\mathbf{z}_1\cdot [\,\mathbf{z}_2\,,\,\mathbf{z}_3\,]\,\right).
\end{equation}
\end{proposition}

The center of the Lie algebra \(\mathbf{H}\) is \(\left\{\,\left(\begin{array}{c}t\\0\end{array}\right)\in \mathbf{H};\,t\in \mathbf{R}\,\right\}\simeq\mathbf{R}\), and (\ref{trivext}) says that  \(\mathbf{H}\) is the trivial central extension of \(\mathfrak{su}(2)\).

  \(\mathbf{R}^3\) being a vector subspace of  \(\mathbf{H}\):
\begin{equation}
\mathbf{R}^3\ni \left(\begin{array}{c}p\\q\\r\end{array}\right)\,\Longleftrightarrow   \left(\begin{array}{c}ir\\q+ip\end{array}\right)=ir+j(q+ip)\in \mathbf{H},\label{action}
\end{equation}
we have the action of \(\mathbf{H}\) on \(\mathbf{R}^3\).

\subsection{ Lie algebra structure on \(\mathbf{H}\otimes U(\mathfrak{g})\) }

Let \((\, \mathfrak{g}\, ,\, \bigl[\quad ,\quad\bigr]_{\mathfrak{g}}\,) \) be a complex Lie algebra.     Let 
\(U(\mathfrak{g})\) be the enveloping algebra of \(\mathfrak{g} \).     Let  \(\mathfrak{g}^{\mathbf{H}}= \mathbf{H}\otimes U(\mathfrak{g})\) and define the following bracket on \(\mathfrak{g}^{\mathbf{H}}\) :  
 \begin{equation}\label{brac}
 [\,\mathbf{z} \otimes X\, , \,\mathbf{w}  \otimes Y\,]_{\mathfrak{g}^{\mathbf{H}}} 
  =  (\mathbf{z}\cdot\mathbf{w})\otimes\,(XY) - (\mathbf{w} \cdot \mathbf{z}) \otimes (YX) \,
  \end{equation}
 for \(X ,\, Y \in U(\mathfrak{g})\) and \(\mathbf{z},\,\mathbf{w}\in \,\mathbf{H}\, \).

 By the quaternion number notation every element of \(\mathbf{H}\otimes \mathfrak{g}\) may be written as \(X+jY\) with \(X,Y \in \mathfrak{g}\).
Then the above definition is equivalent to 
\begin{eqnarray}
\bigl[\,X_1+jY_1\,,\,X_2+jY_2\,\bigr]_{\mathfrak{g}^{\mathbf{H}}} &=&\,[X_1,X_2]_{\mathfrak{g}}\,-\,(\,\overline Y_1Y_2\,-\,\overline Y_2Y_1\,)\nonumber\\[0.2cm]
&&\,+\,j\left( \, \overline X_1Y_2-Y_2X_1+Y_1X_2-\overline X_2Y_1\,\right)\,,
\end{eqnarray}
where \(\overline X\) is the complex conjugate of \(X\).   

\begin{proposition}\label{qtf}
The bracket  \( \bigl[\,\cdot\,,\,\cdot\,\bigr]_{\mathfrak{g}^{\mathbf{H}}}\) defines a Lie algebra structure on \(\mathbf{H}\otimes U(\mathfrak{g})\) .
 \end{proposition} 
 In fact the bracket defined in (\ref{brac}) satisfies the antisymmetry equation and the Jacobi identity.
 
\begin{definition}
The Lie algebra 
 \(\left(\mathfrak{g}^{\mathbf{H}}=
 \mathbf{H}\otimes U(\mathfrak{g})\,,\,\bigl[\,\quad,\quad \bigr]_{\mathfrak{g}^{\mathbf{H}}}\,\right)\) is called the  {\it quaternification of the Lie algebra} \(\mathfrak{g}\) .
\end{definition}

\section{Analysis on $\mathbf{H}$}

In this section we shall review the analysis of the Dirac operator on $\mathbf{H} \simeq \mathbf{C}^2$.   The general references are \cite{B-D-S} and \cite{ G-M}, and we follow the calculations  developed in  \cite{Ko1}, \cite{Ko2} and \cite{ Ko3}.

\subsection{Harmonic polynomials}

The Lie group SU(2) acts on $\mathbf{C}^2$ both from the right and from the left.
Let dR(g) and dL(g) denote respectively the right and the left infinitesimal actions of the Lie algebra
 $\mathfrak{su}(2)$.
We define the following vector fields on $\mathbf{C}^2$:
\begin{equation}
\theta _i = dR \left(\frac12e_i \right) ,\quad
 \tau _i = dL \left(\frac12e _i \right) ,
\qquad i = 1,2,3 , \label{lraction}
\end{equation}
where \(\{e_i;\,i=1,2,3\}\) is the normal basis of \(\mathfrak{su}(2)\), (\ref{3basis}) .   
Each of the triple $\theta_i(z)$, $i=1,2,3$,  and $\tau_i(z)$, $i=1,2,3$, gives a basis of the vector fields on 
the three sphere $\{|z|= 1\}\simeq S^3$. 

It is more convenient to introduce the following vector fields:
\begin{eqnarray}
e_+ &=& -z_2 \frac{\partial}{\partial \Bar{z_1}} +z_1 \frac{\partial}{\partial \Bar{z_2}} 
 =\theta_1 - \sqrt{-1} \theta_2 ,\label{triple1}\\
e_- &=& - \Bar{z_2} \frac{\partial}{\partial z_1} + \Bar{z_1} \frac{\partial}{\partial z_2}
 = \theta_1 + \sqrt{-1} \theta_2 ,\label{triple2}\\
\theta &=& z_1 \frac{\partial}{\partial z_1} + z_2 \frac{\partial}{\partial z_2}
- \Bar {z_1} \frac{\partial}{\partial \Bar{z_1}} - \Bar{z_2} \frac{\partial}{\partial \Bar{z_2}}
 = 2\sqrt{-1}\,\theta_3 .\label{triple3}\\
\Hat{e}_+ &=&- \Bar{z_1} \frac{\partial}{\partial \Bar{z_2}} + z_2 \frac{\partial}{\partial z_1}
 =\tau_1 - \sqrt{-1} \tau_2 ,\\
\Hat{e}_- &=& \Bar{z_2} \frac{\partial}{\partial \Bar{z_1}} - z_1 \frac{\partial}{\partial z_2}
 =\tau_1 + \sqrt{-1} \tau_2,\\
\Hat{\theta} &=& z_2 \frac{\partial}{\partial z_2} + \Bar{z_1} \frac{\partial}{\partial \Bar{z_1}}
- \Bar {z_2} \frac{\partial}{\partial \Bar{z_2}} - z_1 \frac{\partial}{\partial z_1}
 = 2\sqrt{-1}\,\tau_3 . 
\end{eqnarray}
We have the commutation relations;
\begin{eqnarray}
[\theta , e_+] = 2e_+, \,& \, [\theta , e_-] = -2e_-, \,& \, [e_+ ,e_-]=- \theta .\\[0,2cm]
[\hat{\theta} , \hat{e}_+] = 2\hat{e}_+, \,&\, [\hat{\theta} , \hat{e}_-] = -2\hat{e}_-,\, &\,
 [\hat{e}_+ ,\hat{e}_-]=- \hat{\theta} . 
\end{eqnarray}

Both Lie algebras spanned by \( (e_+, e_-, \theta) \) and  \( (\hat e_+, \hat e_-, \hat \theta) \) 
 are isomorphic to  $\mathfrak{sl}(2, \mathbf{C})$. 
 
In the following we denote a function $f(z, \Bar{z})$ of variables $z, \Bar{z}$ simply by $f(z)$.
   For \(m = 0, 1, 2, \cdots\), and \( l, k = 0, 1, \cdots , m\),  we define the polynomials:
\begin{eqnarray}
v ^{k} _{(l,m-l)} &=& (e_-)^k z^{l}_{1} z^{m-l}_{2}.\\[0.2cm]  
w ^{k} _{(l,m-l)} &=& (\Hat{e}_-)^k z^{l}_{2} \Bar{z}^{m-l}_{1}.
\end{eqnarray}
Then \(v ^{k} _{(l,m-l)}\) and \(w ^{k} _{(l,m-l)}\)  are harmonic polynomials on $\mathbf{C}^2$;   
\[\Delta v ^{k} _{(l,m-l)}=\Delta w ^{k} _{(l,m-l)}=0\,,\]
   where \(\Delta
= \frac{\partial ^2}{\partial z _1 \partial \Bar{z}_1} + \frac{\partial ^2}{\partial z _2 \partial \Bar{z}_2}
\).    

\noindent  \(\left\{\,\frac{1}{\sqrt{2}\pi}v ^{k} _{(l,m-l)}\, ;  \,m = 0, 1, \cdots,\,  0\leq k,l\leq m\,\right\} \)  forms a \(L^2(S^3)\)-complete orthonormal basis of the space of harmonic polynomials,  as well as   \(\left\{\,\frac{1}{\sqrt{2}\pi}w ^{k} _{(l,m-l)}\, ;  \,m = 0, 1, \cdots,\,  0\leq k,l\leq m\,\right\} \).
 
\begin{proposition}\label{a}
\begin{eqnarray} 
e_+ v ^{k} _{(l,m-l)}  &=&- k(m - k +1) v ^{k-1} _{(l,m-l)} , \notag \\[0.2cm]
e_- v ^{k} _{(l,m-l)}  &= &v ^{k+1} _{(l,m-l)},\\[0,2cm]
\theta v ^{k} _{(l,m-l)}  &=& (m - 2k) v ^{k}_{(l,m-l)}\, . \notag
\end{eqnarray}
\begin{eqnarray} 
\Hat{e}_+ w ^{k} _{(l,m-l)} &=&- k(m - k +1) w ^{k-1} _{(l,m-l)}, \notag \\[0.2cm]
\Hat{e}_-w ^{k} _{(l,m-l)} &=& w ^{k+1} _{(l,m-l)},\\[0.2cm]
\Hat{\theta} w ^{k} _{(l,m-l)}  &=& (m - 2k) w ^{k} _{(l,m-l)}. \notag
\end{eqnarray}
\end{proposition}
Therefore the space of harmonic polynomials on $\mathbf{C}^2$ is decomposed by the right action of SU(2) into $\sum_{m}\sum^{m}_{l=0} H_{m, l}$.
Each $H_{m, l} = \sum^{m}_{k=0} \mathbf{C }\,v^{k}_{(l, m-l)}$ gives an (m+1) dimensional irreducible 
representation of \(SU(2)\) with the highest weight $\frac{m}{2}$, \cite{T}.   

 We have the following relations.
\begin{eqnarray}
 w^{k} _{(l,m-l)}&=& (-1)^k\frac{l!}{(m-k)!}\,v^{m-l}_{(k,m-k)}\, ,\\[0.2cm]
\overline{v^{k} _{(l,m-l)}}&=& (-1)^{m-l-k}\frac{k!}{(m-k)!}v^{m-k} _{(m-l,l)}\, .\label{vform}
\end{eqnarray}

\subsection{Harmonic spinors}

\(\Delta=\mathbf{H}\otimes \mathbf{C}=\mathbf{H}\oplus \mathbf{H}\) gives an irreducible complex representation of the Clifford algebra \({\rm Clif}( \mathbf{R}^4)\): 
\[ {\rm Clif } ( \mathbf{R} ^4)\otimes  \mathbf{C}\,\simeq\, {\rm End} (\Delta)\,.\]
\(\Delta\) decomposes into irreducible representations \(\Delta^{\pm}=\mathbf{H}\) of \({\rm Spin}(4)\).   Let 
\(S= \mathbf{C} ^2\times \Delta\) be the trivial spinor bundle on \( \mathbf{C} ^2\).   The corresponding bundle 
\(S^+= \mathbf{C} ^2\times \Delta^+\) ( resp.  \(S^-= \mathbf{C} ^2\times \Delta^-\) )  is called the even ( resp.  odd ) spinor bundle and the sections are called even ( resp. odd ) spinors.  The set of even spinors or odd spinors on a set \(M\subset  \mathbf{C} ^2\) is nothing but the smooth functions on \(M\) valued in \( \mathbf{H}\):
 \begin{equation}
Map(M, \mathbf{H})\,=\, C^{\infty}(M, S^+)\,.
 \end{equation}

 The Dirac operator is defined by
\begin{equation}
\mathcal{D} = c \circ d
\end{equation}
where $d : S \rightarrow S \otimes T^{*} \mathbf{C}^2 \simeq S \otimes T \mathbf{C}^2$
is the exterior differential and $c: S \otimes T \mathbf{C}^2 \rightarrow S$ is the bundle 
homomorphism coming from the Clifford multiplication.
By means of the decomposition $S = S^{+} \oplus S^{-}$ the Dirac operator has 
the chiral decomposition:
\begin{equation}
\mathcal{D} = 
\begin{pmatrix}
0 & D^{\dagger} \\
D & 0
\end{pmatrix}
: C^{\infty}(\mathbf{C}^2, S^{+} \oplus S^{-}) \rightarrow C^{\infty}(\mathbf{C}^2, S^{+} \oplus S^{-}).
\end{equation}
We find that $D$ and \(D^{\dagger}\) have the following coordinate expressions;
\begin{equation}\label{Dirac}
D =  \begin{pmatrix} \frac{\partial}{\partial z_1} & - \frac{\partial}{\partial \Bar{z_2}} 
\\ \\ \frac{\partial}{\partial z_2} & \frac{\partial}{\partial \Bar{z_1}} \end{pmatrix} , 
\quad
D^{\dagger} = \begin{pmatrix} \frac{\partial}{\partial \Bar{z_1}} & \frac{\partial}{\partial \Bar{z_2}} 
\\ \\ - \frac{\partial}{\partial z_2} & \frac{\partial}{\partial z_1} \end{pmatrix}.
\end{equation}
 An even (resp. odd) spinor $\varphi$ is called a {\it harmonic spinor} if $D \varphi = 0$ 
 ( resp. \(D^{\dagger} \varphi = 0\) ).

We shall introduce a set of harmonic spinors which, restricted to \(S^3\), forms a 
complete orthonormal basis of \(L^{2}(S^3, S^{+})\) .

Let $\nu$ and $\mu$ be vector fields on \(\mathbf{C}^2\) defined by
\begin{equation}
\nu = z_1 \frac{\partial}{\partial z_1} + z_2 \frac{\partial}{\partial z_2} \ , \ \ \
\mu = z_2 \frac{\partial}{\partial z_2} + \Bar{z_1} \frac{\partial}{\partial \Bar{z_1}} \ .
\end{equation}
Then the radial vector field is defined by
\begin{equation}
\frac{\partial}{\partial n} = \frac{1}{2 |z|}(\nu + \Bar{\nu}) =\frac{1}{2 |z|}(\mu + \Bar{\mu}).\label{radial}
\end{equation}

We shall denote by \(\gamma\) the Clifford multiplication of the radial vector $\frac{\partial}{\partial n}\,$ , (\ref{radial}).   
 $\gamma$ changes the chirality:
\begin{align*}
\gamma : S^+ \oplus S^- \longrightarrow S^- \oplus S^+ ; \ \ \ \gamma^2 = 1.
\end{align*}
The matrix expression of $\gamma$ becomes as follows:
\begin{equation}
\gamma  | S^+ = \frac{1}{|z|} 
\begin{pmatrix} \Bar{z_1} & -z_2 \\ \Bar{z_2} & z_1 \end{pmatrix}, \quad
\gamma | S^- = \frac{1}{|z|} 
\begin{pmatrix} z_1 & z_2 \\ -\Bar{z_2} & \Bar{z_1} \end{pmatrix}.
\label{gamma}
\end{equation}
In the sequel we shall write $\gamma_+$ (resp. $\gamma_-$) for $\gamma | S^+$ 
(resp. $\gamma | S^+$).\\

\begin{proposition}
The Dirac operators $D$ and \(D^{\dagger}\) have the following polar decompositions:
\begin{eqnarray*}
D &=& \gamma_+ \left( \frac{\partial}{\partial n} - \Do \right) ,\\[0.2cm]
D^\dagger &=& \left( \frac{\partial}{\partial n} + \Do + \frac{3}{2|z|} \right)\gamma_- \,,
\end{eqnarray*}
where the tangential (nonchiral) Dirac operator $\Do$ is given by
\begin{equation*}
\Do = - \left[ \sum^{3} _{i = 1} \left( \frac{1}{|z|} \theta_i \right) \cdot \nabla_{\frac{1}{|z|} \theta_i} \right]
= \frac{1}{|z|} 
\begin{pmatrix}
-\frac{1}{2} \theta & \,e_+ \\[0.2cm]
-e_- &\, \frac{1}{2} \theta
\end{pmatrix}.
\end{equation*}
\end{proposition}

\begin{proof}
In the matrix expression (\ref{Dirac}) of $D$ and $D^\dagger$, we have \(\frac{\partial}{\partial z_1} = \frac{1}{|z|^2} (\Bar{z_1} \nu - z_2 e_-)\)  etc., 
and we have the 
desired formulas.
\end{proof}
The tangential Dirac operator on the sphere \(S^3 = \{|z| = 1 \}\);
 \begin{equation*}
 \Do | S^3 : C^{\infty} (S^3, S^+) \longrightarrow C^{\infty} (S^3, S^+)
 \end{equation*}
is a self adjoint elliptic differential operator.\\
We put, for $m = 0,1,2, \cdots ; l = 0,1, \cdots , m$ and $k=0,1, \cdots , m+1$,
\begin{eqnarray}\label{basespinor}
\phi^{+(m,l,k)} (z) &=& \sqrt{\frac{(m+1-k)!}{k!l!(m-l)!}} \begin{pmatrix} k v^{k-1} _{(l, m-l)}\\ \\ -v^{k}_{(l, m-l)} \end{pmatrix},\\ \notag \\
\phi^{-(m,l,k)} (z) &=& \sqrt{ \frac{(m+1-k)!}{k!l!(m-l)!}} \left(\frac{1}{\vert z\vert^2}\right)^{m+2}\begin{pmatrix} w^{k} _{(m+1-l,l)}\\ \\ w^{k}_{(m-l,l+1)} \end{pmatrix}.
\end{eqnarray}
\(\phi^{+(m,l,k)}\)  is a harmonic spinor on \(\mathbf{C}^2\) and  \(\phi^{-(m,l,k)}\) is a harmonic spinor on \(\mathbf{C}^2 \backslash \{0\}\) that is regular  at infinity.   

From Proposition \ref{a}
we have the following 
\begin{proposition}
On $S^3 = \{|z| = 1\}$ we have:
 \begin{eqnarray}
 \Do \phi^{+(m,l,k)} &=& \frac{m}{2} \phi^{+(m,l,k)} \,,\\[0.2cm]
 \Do \phi^{-(m,l,k)} &=& -\frac{m+3}{2} \phi^{-(m,l,k)} \, .
 \end{eqnarray}
The eigenvalues of $\,\Do$ are 
 \begin{equation}
 \frac{m}{2} \,,\quad - \frac{m+3}{2} \, ; \quad  m = 0, 1, \cdots,
 \end{equation}
and the multiplicity of each eigenvalue is equal to $(m+1)(m+2)$.\\
The set of eigenspinors
 \begin{equation}
 \left\{ \frac{1}{\sqrt{2}\pi }\phi^{+(m,l,k)}, \quad \frac{1}{\sqrt{2}\pi }\phi^{-(m,l,k)} \,
 ; \quad m = 0, 1, \cdots , \,  0\leq l\leq  m,\, 0\leq k\leq m+1\right\}
 \end{equation}
forms a complete orthonormal system of $L^2 (S^3, S^+)$.\\
\end{proposition}

The constant for normalization of $\phi ^{\pm(m, l, k)}$ is determined by the integral:
\begin{equation}
\int _{S^3} |z^a _1 z^b _2|^2 d \sigma = 2 \pi^2 \frac{a! b!}{(a + b + 1)!}\, ,
\end{equation}
where $\sigma$ is the surface measure of the unit sphere $S^3=\{\vert z\vert=1\}$:
\begin{equation}
\int _{S^3} d \sigma_3 = 2 \pi ^2 .
\end{equation}

\subsection{Spinors of  Laurent polynomial type}

If \(\varphi\) is a harmonic spinor on 
 \(\mathbf{C}^2\setminus\{0\}\) then we have the expansion
 \begin{equation}
 \varphi(z)=\sum_{m,l,k}\,C_{+(m,l,k)}\phi^{+(m,l,k)}(z)+\sum_{m,l,k}\,C_{-(m,l,k)}\phi^{-(m,l,k)}(z),\label{Laurentspinor}
 \end{equation}
that is uniformly convergent on any compact subset of \(\mathbf{C}^2\setminus\{0\}\). 
The coefficients \(C_{\pm(m,l,k)}\) are given by the formula:
\begin{equation}\label{coef}
C_{\pm(m,l,k)}=\,\frac{1}{2\pi^2}\int _{S^3}\, \langle \varphi,\,\phi^{\pm(m,l,k)}\rangle\,d\sigma,
\end{equation}
where \(\langle\,,\,\rangle\) is the inner product of \(S^+\).    

\begin{lemma} \label{trace}
 \begin{eqnarray}
 \int_{S^3}\,tr\,\varphi \,d\sigma&=&4\pi^2Re.C_{+(0,0,1)},\label{coefficient1}\\[0.2cm]
  \int_{S^3}\,tr\,J\varphi \,d\sigma&=&4\pi^2Re.C_{+(0,0,0)}\nonumber.
  \end{eqnarray}
\end{lemma}
The formulas follow from (\ref{coef}) if we take
\(\phi^{ +(0,0,1)}=\left(\begin{array}{c}1\\0\end{array}\right)\) and \(J=\phi^{ +(0,0,0)}=\left(\begin{array}{c}0\\-1\end{array}\right)\),

\begin{definition}
\begin{enumerate}
\item
We call the series (\ref{Laurentspinor})  {\it a spinor of  Laurent polynomial type} if only finitely many coefficients \(C_{\pm (m,l,k)}\)  are non-zero .   The space of spinors of Laurent polynomial type is denoted by   \(\mathbf{C}[\phi ^{\pm}] \).  
 \item
   For a spinor of Laurent polynomial type \(\varphi\) we call the vector \({\rm res}\,\varphi=\,\left(\begin{array}{c}-C_{-(0,0,1)}\\C_{-(0,0,0)}\end{array}\right)\) the residue at \(0\) of \(\varphi\).  
   \end{enumerate}
\end{definition} 

 We have the residue formula, \cite{Ko3}.
\begin{equation}
{\rm res}\,\varphi=\frac{1}{2\pi^2}\int_{S^3}\gamma_+(z)\varphi(z)\sigma(dz). \end{equation}

\begin{remark}\par
\medskip
To develop the spinor analysis on the 4-sphere \(S^4\) we patch two local coordinates \(\mathbf{C}^2_z\) and \(\mathbf{C}^2_w\) together by the inversion \(w=-\frac{\overline z}{|z|^2}\).   This is a conformal transformation with the conformal weight \(u=-\log |z|^2\).   An even spinor on a subset \(U\subset S^4\) is a pair of \(\phi\in C^{\infty}(U\cap \mathbf{C}^2\times \Delta)\) and \(\widehat\phi\in C^{\infty}(U\cap \widehat{\mathbf{C}}^2\times \Delta)\) such that \(\widehat\phi(w)=\overline{|z|^3(\gamma_+\phi)(z)}\) for \(w=-\frac{\overline z}{|z|^2}\).    
Let \(\varphi\) be a spinor of Laurent polynomial type on \(\mathbf{C}^2\setminus {0}=\widehat{\mathbf{C}}^2\setminus\widehat{0}\).   The coefficient  \(C_{\pm(m,l,k)}\)  of  \(\varphi\) and the coefficient  \(\widehat{C}_{\pm(m,l,k)}\,\) of \(\widehat\varphi\) are related by the formula:    
\begin{equation}
\widehat{C}_{-(m,l,k)}=\,\overline C_{+(m,l,k)},\qquad \widehat{C}_{+(m,l,k)}=\,\overline C_{-(m,l,k)}.
\end{equation}
\end{remark}
\begin{proposition}
The residue of 
\(\widehat{\varphi}\) is related to the trace of \(\varphi\,\), (\ref{trace}),  by
\begin{equation}
\,res\, \widehat\varphi\,=\,\frac{1}{2\pi^2}\int _{S^3}\,\widehat \varphi\,d\sigma\,=\,\left(\begin{array}{c}\overline C_{+(0,0,1)}\\ - \overline C_{+(0,0,0)}\end{array}\right)\,.
\end{equation}
\end{proposition}

\subsection{ Algebraic generators of   \(\mathbf{C}[\phi ^{\pm}] \) }

In the following we show that  \( \mathbf{C}[\,\phi^{\pm}\,]\) restricted to \(S^3\) becomes an algebra.         
The multiplication of two harmonic polynomials on \(\mathbf{C}^2\)  is not harmonic but its restriction to  \(S^3\) is again the restriction to \(S^3\) of some harmonic polynomial.    We shall see that this yields the fact that \( \mathbf{C}[\phi^{\pm}\,]\), restricted to \(S^3\),  becomes an associative subalgebra of \(S^3\mathbf{H}\).   
Before we give the proof we look at examples that convince us of the necessity of the restriction to \(S^3\) .

{\bf Example 1}\\
\(\phi^{+(1,0,1)}\cdot \phi^{-(0,0,0)}\,\) is decomposed to the sum
\[
\phi^{+(1,0,1)}(z)\cdot \phi^{-(0,0,0)}(z)=
\frac{1}{|z|^4}(\,\frac{2}{3}\phi^{+(2,0,1)}(z)+\frac{\sqrt{2}}{3}\phi^{+(2,1,2)}(z)\,)+
\frac{1}{|z|^2}
\frac{1}{2}\phi^{+(0,0,1)}(z),\]
which is not in \(\mathbf{C}[\phi^{\pm}]\).   But the restriction to \(S^3\) is 
\[
\,\frac{2}{3}\phi^{+(2,0,1)}+\frac{\sqrt{2}}{3}\phi^{+(2,1,2)}\,+\frac{1}{2}\phi^{+(0,0,1)}+\frac{1}{6}\phi^{-(1,1,1)}+\frac{1}{3\sqrt{2}}\phi^{-(1,0,0)}\,\in \mathbf{C}[\phi^{\pm}]\vert_{ S^3}.\]
See the table at the end of this subsection.

We start with the following facts.
\begin{enumerate}
\item
We have the product formula for the harmonic polynomials \(v^k_{(a.b)}\).   
\begin{equation}\label{product}
v ^{k_1} _{(a_1,b_1)} v ^{k_2} _{(a_2,b_2)}=\sum_{j=0}^{a_1+a_2+b_1+b_2}C_j\vert z\vert^{2j}\,v ^{k_1+k_2-j} _{(a_1+a_2-j,\,b_1+b_2-j)} \, 
\end{equation}
for some rational numbers  \(C_j=C_j(a_1,a_2,b_1,b_2,k_1,k_2)\), see  Lemma 4.1 of \cite{Ko1}.   
\item
Let \(k=k_1+k_2\), \(a=a_1+a_2\) and \(b=b_1+b_2\).    The above (\ref{product}) yields that, restricted to \(S^3\), the harmonic  polynomial \(v^k_{(a,b)}\) is equal to  a constant multiple of 
\(\,v^{k_1}_{(a_1,b_1)} v^{k_2}_{(a_2,b_2)}\) modulo a linear combination of polynomials \(v^{k-j}_{(a-j,b-j)}\,\), \(1\leq j\leq min(k,a,b)\).

\item
\(\left(\begin{array}{c}v^k_{(l,m-l)}\\0\end{array}\right)\) and \(\left(\begin{array}{c}0\\v^{k+1}_{(l,m-l)}\end{array}\right)\) are written by linear combinations of \(\phi^{+(m,l,k+1)}\) and \(\phi^{-(m-1,k,l)}\).
\item
Therefore the product of two spinors \(\phi^{\pm (m_1,l_1,k_1)}\cdot \phi^{\pm (m_2,l_2,k_2)}\) belongs to \( \mathbf{C}[\phi^{\pm}]\vert_{S^3}\) .  \( \mathbf{C}[\phi^{\pm}]\vert_{S^3}\)  becomes an associative algebra.   
\item
 \(\phi^{\pm(m,l,k)}\) is written by a linear combination of the products  \(\,\phi^{\pm(m_1,l_1,k_1)}\cdot \phi^{\pm(m_2,l_2,k_2)}\) for \(0\leq m_1+m_2\leq m-1\,\),  \(0\leq l_1+l_2\leq l\) and \(0\leq k_1+k_2\leq k\) .
\end{enumerate}
Hence 
we find that the algebra  \( \mathbf{C}[\phi^{\pm}]\vert_{S^3}\)  is generated by the following \(I,\,J,\,\kappa,\,\mu\):
\begin{align}\label{generate}
I=\phi^{ +(0,0,1)}=\left(\begin{array}{c}1\\0\end{array}\right),\,&\quad 
J=\phi^{ +(0,0,0)}=\left(\begin{array}{c}0\\-1\end{array}\right),\nonumber \\[0.2cm]
\kappa= \phi^{ +(1,0,1)}=\left(\begin{array}{c}z_2\\-\overline z_1\end{array}\right),&\quad 
\mu=\phi^{-(0,0,0)}=\left(\begin{array}{c}z_2\\ \overline z_1\end{array}\right) .
\end{align}
 The others are generated by these basis,   For example,
 \begin{align}
\lambda=  \phi^{ +(1,1,1)}=\left(\begin{array}{c} z_1\\ \overline z_2\end{array}\right) =-\kappa J\,
, & \quad 
\nu= \phi^{-(0,0,1)}= \left(\begin{array}{c}-z_1\\  \overline z_2\,
\end{array}\right) =-\mu J , \nonumber\\[0.2cm]
 \phi^{ +(1,0,0)}=\sqrt{2}\left(\begin{array}{c}0\\-z_2\end{array}\right)
 =\frac{1}{\sqrt{2}}J(\kappa+\mu), &\quad
\phi^{+(1,0,2)}=\sqrt{2}\left(\begin{array}{c} \overline z_1\\ 0 \end{array}\right)=\frac{1}{\sqrt{2}}J(\mu-\kappa)\,
,\nonumber\\[0.2cm]
\phi^{ +(1,1,2)}=\sqrt{2}\left(\begin{array}{c}-\overline z_2\\0\end{array}\right)=-\frac{1}{\sqrt{2}}J(\lambda+\nu) , &\quad
 \phi^{ +(1,1,0)}=\sqrt{2}\left(\begin{array}{c}0\\-z_1\end{array}\right)
 =\frac{1}{\sqrt{2}}J(\lambda-\nu),\nonumber \\[0.2cm]
\phi^{ -(1,0,0)}=\sqrt{2}\left(\begin{array}{c}z_2^2\\z_2
\overline z_1\end{array}\right)=\frac{1}{\sqrt{2}}\nu J(\kappa+\mu) , &\quad
\phi^{ -(1,1,0)}=\sqrt{2}\left(\begin{array}{c}z_2\overline z_1\\
\overline z_1^2\end{array}\right)=\frac{1}{\sqrt{2}}\mu J(\mu-\kappa) , \nonumber\\[0.2cm]
\phi^{ -(1,1,2)}=\sqrt{2}\left(\begin{array}{c}-z_1\overline z_2\\
\overline z_2^2\end{array}\right)=\frac{1}{\sqrt{2}}\nu J(\lambda+\nu) , &\quad
\phi^{-(1,0,2)}=\sqrt{2}\left(\begin{array}{c}  z_1^2\\ -z_1\overline z_2 \end{array}\right)=\frac{1}{\sqrt{2}}\mu J(\lambda-\nu)\,
\nonumber\\[0.2cm]
\phi^{ -(1,0,1)}=\left(\begin{array}{c}-2z_1 z_2\\
\vert z_2\vert^2-\vert z_1\vert^2\end{array}\right)=
\frac{\nu}{2}(\kappa+\mu&+ J(\lambda-\nu)) , \quad\,\nonumber\\[0.2cm]
\phi^{-(1,1,1)}=\left(\begin{array}{c} \vert z_2\vert^2-\vert z_1\vert^2\\ 2\overline z_1\overline z_2 \end{array}\right)=
\frac{\mu}{2}(-\kappa+\mu&+J(\lambda+\nu))\, . \quad\,
\nonumber
 \end{align}
 
\medskip

\subsection{2-cocycle on   \(S^3\mathbf{H} \)} 

Let \(S^3\mathbf{H}=Map(S^{3}, \mathbf{H})\,=\,C^{\infty}(S^3, S^+)\) be the set of smooth even spinors on \(S^3\).   
We define the Lie algebra structure on  \(S^3\mathbf{H}\)  after  (\ref{Liebr}), that is, for even spinors   $\phi _1 = \begin{pmatrix} u_1\\ v_1 \end{pmatrix}$ and $\phi _2 = \begin{pmatrix} u_2\\ v_2 \end{pmatrix}$,   we have the Lie bracket
\begin{equation}
\bigl[\,\phi _1\, , \,\phi _2\,\bigr]=
 \begin{pmatrix} \,v_1 \Bar{v} _2 - \Bar{v}_1 v_2 \,\\[0.2cm]
\,(u _ 2 - \Bar {u} _2 ) v _1 - (u _1 - \Bar {u} _1) v_2\, \end{pmatrix} .
\end{equation}

For a \(\varphi=\left(\begin{array}{c}u\\ v\end{array}\right)\in S^3\mathbf {H}\), we put 
\[
\Theta\,\varphi\,=\,\left(\begin{array}{cc}\frac{1}{2\sqrt{-1}}\theta& 0\\[0.3cm] 0&-\frac{1}{2\sqrt{-1}}\overline \theta\end{array}\right)\left(\begin{array}{c}u\\ v\end{array}\right)\,=\,\frac{1}{2\sqrt{-1}}\,\left(\begin{array}{c}\,\theta\, u\\[0.3cm] \,\theta\, v\end{array}\right).\]

\begin{lemma}~~ For any  \(\phi,\,\psi \in S^3\mathbf{H}\),  we have 
\begin{eqnarray}
\Theta\,(\phi\cdot\psi\,)\,&=&\,(\Theta\,\phi)\cdot \,\psi\,+\,\phi\cdot\,(\Theta\,\psi)\,.
\label{leibnitz}
\\[0.2cm]
\int_{S^3}\,\Theta\,\varphi\,d\sigma\,&=&\,0\,.\label{tracetheta}
\end{eqnarray}
\end{lemma}
The second assertion follows from the fact 
\[\int_{S^3}\,\theta f\,d\sigma\,=\,0\,,\]
for any function \(f\) on \(S^3\).

\begin{proposition}~~
\begin{eqnarray*}
2\sqrt{-1}\,\Theta\,\phi^{+(m,l,k)}\,&=&\,\,\frac{m(m+1-k)}{m+1}\,\,\phi^{+(m,l,k)}\,+\,2(-1)^l\,\frac{\sqrt{k(m+1-k)}}{m+1}\,\phi^{-(m-1,k-1,l)}\,,\\[0.3cm]
2\sqrt{-1}\,\Theta\,\phi^{-(m,l,k)}\,&=&\,(m-2l)\frac{m+3}{m+2}\,\phi^{-(m,l,k)}\,+\,2(-1)^k\,
\frac{\sqrt{(l+1)(m+1-l)}}{m+2}\,\phi^{+(m+1,k,l+1)}\,,
\end{eqnarray*}
 on \(S^3\).
\end{proposition}
\par\medskip
Now we shall introduce a non-trivial 2-cocycle on \(S^3\mathbf{H}\,\) .
\begin{definition}
For \(\phi_1\) and \(\phi_2\in S^3\mathbf{H}\,\), we put 
\begin{equation}
c(\phi_1,\phi_2)\,=\,
\,\frac{1}{2\pi^2}\int_{S^3}\,\,tr\,(\,\Theta \phi_1\cdot \phi_2\,)\, d\sigma. 
  \end{equation}
\end{definition}

  {\bf Example.}\\
\begin{equation}
c(\,\frac{1}{\sqrt{2}}\phi^{+(1,1,2)}\,,\,\frac{\sqrt{-1}}{2}(\phi^{+(1,0,1)}+\phi^{-(0,0,0)})\,)
\,=\frac12.
\end{equation}

\begin{proposition}\label{2cycle}
\(c\) defines a non-trivial 2-cocycle on the algebra \(S^3\mathbf{H}\,\).    That is, \(c\) satisfies the following equations:
\begin{eqnarray}
&&c(\phi_1\,\phi_2)\,=\,-\, c(\phi_2,\,\phi_1)\,,\label{asym}\\[0.3cm]
&&c(\phi_1\cdot\phi_2\,,\,\phi_3)+c(\phi_2\cdot\phi_3\,,\,\phi_1\,)+c(\phi_3\cdot\phi_1\,,\,\phi_2\,)=0 .\label{cycle}
\end{eqnarray}
And there is no 1-cochain \(b\) such that \(c(\phi_1.\phi_2)=b(\,[\phi_1,\phi_2]\,)\).
\end{proposition}
{\it Proof.}~~
By (\ref{tracetheta}) and the Leibnitz rule (\ref{leibnitz}) we have
\begin{eqnarray*}
0&=&\, \int_{S^3}\,\,tr\,(\,\Theta\,(\phi_1\cdot\phi_2)\,)\,d\sigma\,=\,\int_{S^3}\,tr\,\left(\,\Theta\,\phi_1\,\cdot\phi_2\,\right)d\sigma\,+\,\int_{S^3}\,tr\,\left(\,\phi_1\cdot\,\Theta\,\phi_2\,\right) d\sigma
 \end{eqnarray*}
Hence
 \[ c(\phi_1\,,\,\phi_2\,)\,+\,c(\,\phi_2\,,\,\phi_1\,)\,=0\,.
 \]
  The following calculation proves (\ref{cycle}).   
\begin{eqnarray*}
c(\phi_1\cdot\phi_2\,,\,\phi_3)&=&\, \int_{S^3}\,\,tr\,(\,\Theta(\,\phi_1\cdot\phi_2\,)\cdot\,\phi_3\,)\,d\sigma
 \\[0.2cm]
&=&\, \int_{S^3}\,\,tr\,(\,\Theta\phi_1\cdot \phi_2\cdot \phi_3\,)d\sigma \,+\,\, \int_{S^3}\,\,tr\,(\,\Theta\phi_2\cdot\,\phi_3\,\cdot\phi_1\,)d\sigma \\[0.2cm]
&=&\,c(\phi_1\,,\,\phi_2\cdot\phi_3\,)+c(\phi_2\,,\,\phi_3\cdot\phi_1\,)
=\,- c(\phi_2\cdot\phi_3\,,\phi_1\,)\,-\,c(\phi_3\cdot\phi_1\,,\,\phi_2)
.
\end{eqnarray*}  
Suppose now that \(c\) is the coboundary of a 1-cochain 
\(b:\, S^3\mathbf{H}\longrightarrow \mathbf{C}\).    Then 
\begin{equation}
c(\phi_1,\,\phi_2)=(\delta\,b)(\phi_1,\,\phi_2)\,=\,b(\,[\phi_1,\,\phi_2]\,)\nonumber
\end{equation}
for any \(\phi_1,\phi_2\in S^3\mathbf{H}\).   
Take 
\(\phi_1
=\,\frac{1}{\sqrt{2}}\phi^{+(1,1,2)}\,
=\left(\begin{array}{c} -\overline z_2 \\  0\end{array}\right)\)
and \(\phi_2=\frac{\sqrt{-1}}{2}( \phi^{+(1,0,1)}+\phi^{-(0,0,0)})=\left(\begin{array}{c}\sqrt{-1}z_2\\ 0\end{array}\right)\) .   Then \([\,\phi_1,\,\phi_2\,]=0\),   
so 
\((\delta b)(\phi_1,\phi_2)=0\).
But \(c(\phi_1,\phi_2)=\frac12\,\).   Therefore \(c\) can not be a coboundary.
  \hfill\qed

\subsection{
Calculations of the 2-cocycle on the basis}

We shall calculate the values  of 2-cocycles \(c \)  for the basis \(\{  \phi ^{\pm(m,l,k)} \}\) of \( \mathbf{C}[\phi ^{\pm}] \).   First we have a lemma that is useful  for the following calculations.

\begin{lemma}~~
\begin{enumerate}
\item
\begin{equation}
\int_{S^3}\,v^{k}_{(a,b)}\,\overline v^{l}_{(c,d)}\,d\sigma\,=\,
2\pi^2\frac{a!b!} {(a+b+1)}\, \frac{k!}{ (a+b-k)!}\,\delta_{a, c}\,\delta_{b, d}\,\delta_{k, l}\,.
\end{equation}

\item
\begin{equation}
\int_{S^3}\,v^{k}_{(a,b)}\, v^{l}_{(c,d)}\,d\sigma\,=(-1)^{b-k}\,2\pi^2
\frac{a!b!} {(a+b+1)}\, \,\delta_{a, d}\,\delta_{b, c}\,\delta_{(a+b-k), l }\,.
\end{equation}

\item
\begin{equation}
\int_{S^3}\,w^{k}_{(a,b)}\,\overline w^{l}_{(c,d)}\,d\sigma\,=\,2\pi^2
\frac{a!b!} {(a+b+1)}\, \frac{k!}{ (a+b-k)!}\,\delta_{a, c}\,\delta_{b, d}\,\delta_{k, l}\,.
\end{equation}
\item
\begin{equation}
\int_{S^3}\,w^{k}_{(a,b)}\, w^{l}_{(c,d)}\,d\sigma\,=(-1)^{b-k}\,2\pi^2
\frac{a!b!} {(a+b+1)}\, \delta_{a, d}\,\delta_{b, c}\,\delta_{(a+b-k), l }\,.
\end{equation}
\end{enumerate}
\end{lemma}

\par\medskip

\begin{lemma}~~
\begin{enumerate}
\item
\begin{equation*}
c(\phi^{\pm(m,l,k)} , \phi^{\pm(p,q,r)}) = 0 \, .
\end{equation*}
\item
\begin{eqnarray*}
c (\phi^{+(m,l,k)}, \sqrt{-1}\phi^{+(p,q,r)})
&= &(-1)^{m-l-k+1} \frac{(m-2k+2)\sqrt{k(m + 2 - k) }}{m+1}\,\delta_{m,p}\,\delta_{l,p-q}\,\delta_{k,\,p-r+2}\\[0.3cm]
&& \quad -\,\frac{(m-2k)(m-k+1)}{m+1}\,\delta_{m,p}\,\delta_{l,q}\,\delta_{k,r}
.
\end{eqnarray*}
\item
\begin{eqnarray*}
c(\phi^{+(m,l,k)}, \sqrt{-1}\phi^{-(p,q,r)})
&=&
(-1)^{k-1} \frac{(m-2k+2)\sqrt{(k-1)k}}{m+1} \, \delta_{m,p+1}\,\delta_{l,\,p-r+1}\,\delta_{k,p-q+2} \\[0.3cm]
&&\quad +\,(-1)^l \,\frac{(m-2k)\sqrt{k(m+1-k)}}{m+1} \,\delta_{m,p+1}\delta_{l,r}\delta_{k,q+1}\,.
\end{eqnarray*}
\item
\begin{eqnarray*}
c(\phi^{-(m,l,k)}, \sqrt{-1}\phi^{+(p,q,r)})
&=&
(-1)^{m+1-l} \frac{(m-2l+1)\sqrt{(m-l+1)(m-l+2)}}{m+2} \, \delta_{m,p-1}\delta_{l,\,p-r+1}\delta_{k,p-q} \\[0.3cm]
&&+(-1)^k \frac{(m-2l+1)\sqrt{(l+1)(m-l+1)}}{m+2} \,\delta_{m,p-1}\delta_{l,r-1}\delta_{k,q}\,.
\end{eqnarray*}
\item
\begin{eqnarray*}
 c(\phi^{-(m,l,k)}, \sqrt{-1}\phi^{-(p,q,r)})
&=&
(-1)^{l-k} \frac{(m-2l+1)\sqrt{l(m-l+1)}}{m+2} \,\delta_{m,p}\delta_{l,p-q+1}\delta_{k,p-r+1} \\[0.3cm]
&&-\, \frac{(m-2l-1)(l+1)}{m+2} \,\delta_{m,p}\delta_{l,q}\delta_{k,r}
\end{eqnarray*}

\end{enumerate}

\end{lemma}

{\it Proof}

Since \(\theta\,v^k_{(a,b)}=(a+b-2k)v^k_{(a,b}\) , we have 
\begin{align*}
c(\phi^{+(m,l,k)} , \sqrt{-1}\phi^{+(p,q,r)}) 
&= \frac{1}{2 \pi^2} \int_{S^3} tr\, [\Theta \phi^{+(m,l,k)} \cdot \sqrt{-1}\phi^{+(p, q, r)}] d \sigma\\[0.3cm]
=\frac
{1}{4 \pi^2} \sqrt{\frac{(m + 1 - k)!}{k! l ! (m - l)!}}&  \sqrt{\frac{(p + 1 - r)!}{r! q ! (p - q)!}}\,\int_{S^3} tr\, \left[
\begin{pmatrix}
k (m-2k+2) v^{k-1}_{(l, m-l)}\\[0.2cm]
-(m-2k) v^{k}_{(l, m-l)}
\end{pmatrix}\cdot
\begin{pmatrix}
r v^{r-1}_{(q, p-q)} \\[0.2cm]
-v^{r}_{(q, p-q)}
\end{pmatrix}
\right] d\sigma\,.
\end{align*}
By the above lemma we obtain the value of \(c(\phi^{+(m,l,k)}, \sqrt{-1}\phi^{+(p,q,r)})\).   The others follow similarly.
 \hfill\qed

  \subsection{Radial derivative on \(S^3\mathbf{H}\)}
 
 We define the following operator \(d_0\) on \(C^{\infty}(S^3)\):
\begin{equation}
d_0\,f (z)= |z|\frac{\partial}{\partial n}f(z)\,=\,\frac{1}{2}(\nu+\Bar{\nu})f(z).\label{d_0}
\end{equation}
For an even spinor \(\varphi=\left(\begin{array}{c}u\\ v\end{array}\right)\) 
we put 
\[
d_0\,\varphi=\left(\begin{array}{c}d_0\,u\\[0.2cm] d_0\,v\end{array}\right).\]
Note that if \(\varphi\in \mathbf{C}[\phi^{\pm}]\) then \(d_0\varphi\in \mathbf{C}[\phi^{\pm}]\).

\begin{proposition}~~
\begin{enumerate}
\item
\begin{equation}
d_0( \phi_1 \cdot \phi_2 ) =(d_0 \phi _1) \cdot \phi_2 + \phi _1 \cdot (d_0 \phi _2)\,.
\label{leibnitz1}
\end{equation}
\item
\begin{equation}
d_0 \phi^{+(m,l,k)} = \frac{m}{2}\, \phi^{+(m,l,k)},\quad 
d_0 \phi^{-(m,l,k)}=\,- \frac{m+3}{2}\, \phi^{-(m,l,k)}.\label{normalder}
\end{equation}
\item
Let \(\varphi=\phi_1\cdots\phi_n\) such that \(\phi_i=\phi^{+(m_i,l_i,k_i)}\) or  \(\,\phi_i=\phi^{-(m_i,l_i,k_i)}\), \(i=1,\cdots,n\).     We put 
\[N=\sum_{i:\,\phi_i=\phi^{+(m_i,l_i,k_i)}}\,m_i\,\,-\,\sum_{i:\,\phi_i=\phi^{-(m_i,l_i,k_i)}}\,(m_i+3).\]
Then 
\begin{equation}
d_0(\varphi)=\,\frac{N}{2}\,\varphi.\label{d0value}
\end{equation}
\item
Let \(\varphi\) be a spinor of Laurent polynomial type:
\begin{equation}
 \varphi(z)=\sum_{m,l,k}\,C_{+(m,l,k)}\phi^{+(m,l,k)}(z)+\sum_{m,l,k}\,C_{-(m,l,k)}\phi^{-(m,l,k)}(z).
 \end{equation}
 Then
 \begin{equation}
 \int_{S^3}\,tr\,(\,d_0\,\varphi\,)\,d\sigma\,=\,0\,.
 \end{equation}
\end{enumerate}
\end{proposition}

{\it Proof.}~~
The formula (\ref{normalder}) follows from the definition (\ref{basespinor}).     The last assertion follows from the fact that the coefficient of \(\phi^{+(0,0,1)} \) in  the Laurent expansion of \(d_0\varphi\)   vanishes.
 \hfill\qed

\begin{definition}\label{homog}
Let  \(\mathbf{C}[\phi^{\pm};\,N]\) be the subspace of \(\mathbf{C}[\phi^{\pm}]\) consisting of those elements that are of homogeneous order \(N\): \(\varphi(z)=|z|^N\varphi(\frac{z}{|z|})\).
\end{definition}
\(\mathbf{C}[\phi^{\pm};\,N]\) is spanned by the spinors  
\(\varphi=\phi_1\cdots\phi_n \) such that each \(\phi_i\) is equal to  \(\phi_i=\phi^{+(m_i,l_i,k_i)}\) or  \(\,\phi_i=\phi^{-(m_i,l_i,k_i)}\) , where \(m_i\geq 0\) and \(0\leq l_i\leq m_i+1, \,0\leq k_i\leq m_i+2\) as before, and such that 
\[N=\sum_{i:\,\phi_i=\phi^{+(m_i,l_i,k_i)}}\,m_i\,\,-\,\sum_{i:\,\phi_i=\phi^{-(m_i,l_i,k_i)}}\,(m_i+3).\]

\(\mathbf{C}[\phi^{\pm}]\) is decomposed into the direct sum of \(\mathbf{C}[\phi^{\pm};\,N]\):
\[
\mathbf{C}[\phi^{\pm}]\,=\, \bigoplus_{N\in \mathbf{Z}}\,\mathbf{C}[\phi^{\pm};\,N]\,.\]
(\ref{d0value}) implies that the eigenvalues of \(d_0\) on \(\mathbf{C}[\phi^{\pm}]\) are \(\left\{\frac{N}{2};\,N\in\mathbf{Z}\,\right\}\) and 
\(\mathbf{C}[\phi^{\pm};\,N]\) is the space of eigenspinors for the eigenvalue \(\frac{N}{2}\).

{\bf Example}
\begin{eqnarray*}
&&\phi^{+(1,0,1)}\cdot\phi^{-(0,0,0)}\,\in \,\mathbf{C}[\,\phi^{\pm};\,-2\,],\\[0.2cm]
&& d_0(\phi^{+(1,0,1)}\cdot\phi^{-(0,0,0)})=-\frac{2}{2}(\phi^{+(1,0,1)}\cdot\phi^{-(0,0,0)}).
\end{eqnarray*}

\begin{proposition}\label{deriv}
\begin{equation}
c(\,d_0\phi_1\,,\phi_2\,)\,+\,c(\,\phi_1\,, d_0\phi_2\,)\,=\,0\,.
\end{equation}
\end{proposition}

In fact,
since \(\,\theta\,d_0\,=\,(\nu-\Bar{\nu})(\nu+\Bar{\nu})=\nu^2-\Bar{\nu}^2=\,d_0\,\theta\,\), we have
\begin{eqnarray*}
0&=&\int_{S^3}\,tr\,(\,d_0(\Theta\phi_1\cdot\phi_2)\,)\,d\sigma=\int_{S^3}\,tr\,(\,(d_0\Theta\phi_1)\cdot \phi_2+\Theta\phi_1\cdot d_0\phi_2\,)\,d\sigma\\[0.2cm]
&=&\int_{S^3}\,tr\,((\Theta\, d_0\phi_1)\cdot\phi_2\,)\,d\sigma+\int_{S^3}\,tr\,(\Theta\phi_1\cdot d_0\phi_2\,)\,d\sigma.\\[0.2cm]
&=& c(d_0\phi_1,\phi_2)+c(\phi_1,d_0\phi_2)
\end{eqnarray*}
   \hfill\qed

\section{Extensions of the Lie algebra   \(\mathbf{C}[\phi ^{\pm}] \otimes U(\mathfrak{g})\)}

  In this section we shall construct a central extention for the 3-dimensional loop algebra \(
Map(S^3,\mathfrak{g}^{\mathbf{H}})=S^3\mathbf{H}\otimes U(\mathfrak{g})\) associated to the above 2-cocycle \(c\), and the central extension of \(\mathbf{C}[\phi^{\pm}]\otimes U(\mathfrak{g})\) induced from it.   Then we shall give the second central extension by adding a derivative to the first extension that acts as the radial derivation.

\subsection{Extension of  \(S^3{\mathfrak{g}}^{\mathbf{H}}\,=S^3\mathbf{H}\otimes U(\mathfrak{g})\,\) }

 From Proposition \ref{qtf} we see that 
   \(S^3{\mathfrak{g}}^{\mathbf{H}}\,=S^3\mathbf{H}\otimes U(\mathfrak{g})\,\) endowed with the following bracket  \([\,,\,]_{S^3{\mathfrak{g}}^{\mathbf{H}}}\) becomes a Lie algebra.
 \begin{equation}
 [\,\phi \otimes X\, , \,\psi  \otimes Y\,]_{ S^3{\mathfrak{g}}^{\mathbf{H}}} =  ( \phi\cdot\psi ) \otimes (XY) \,
- (\psi \cdot \phi ) \otimes (YX)
 , 
\end{equation}
for 
  \(X,Y\in U(\mathfrak{g})\) and \(\phi,\,\psi\,\in S^3\mathbf{H}\, \) .   
 
 We take the non-degenerate invariant symmetric bilinear \(\mathbf{C}\)-valued form \((\,\cdot\,\vert\,\cdot\,)\) on \(\mathfrak{g}\) and extend it to \(U(\mathfrak{g})\).    
 For 
$X = X_1^{l_1} \cdots X_m^{l_m}$ and $Y = Y_1^{k_1} \cdots Y_m^{k_m}$  written by the basis  \(\,X_1,\cdots,X_m,\,Y_1,\cdots,Y_m\,\) 
of \(\mathfrak{g}\), \((X \vert Y)\) is defined by 
 \[(X \vert Y)=tr(ad(X_1^{l_1}) \cdots ad(X_m^{l_m}) ad(Y_1^{k_1}) \cdots ad(Y_m^{k_m})).\]  
Then we define a \(\mathbf{C}\)-valued 2-cocycle on the Lie algebra   \( S^3\mathfrak{g}^{\mathbf{H}}\,\) by 
\begin{equation}
c(\,\phi_1\otimes X\,,\,\phi_2\otimes Y\,)=\,(X \vert Y)\,c(\phi_1,\phi_2).
\end{equation}
The 2-cocycle property follows from the fact \((XY\vert Z)=(YZ\vert X)\) and Proposition  
\ref{2cycle}.
  
 Let \(a\) be an indefinite  number.    There is an extension of the Lie algebra   \( S^3\mathfrak{g}^{\mathbf{H}}\,\) by the 1-dimensional center \(\mathbf{C}a\) associated to the cocycle \(c\).   Explicitly we have the following theorem.
 \begin{theorem}~~ The \(\mathbf{C}\)-vector space 
 \begin{equation}
S^3\mathfrak{g}^{\mathbf{H}}(a)\,=\, (\, S^3\mathbf{H} \otimes U(\mathfrak{g})\,)\oplus (\mathbf{C}a),
\end{equation}
endowed with the following bracket becomes a Lie algebra.    
 \begin{eqnarray}
 [\,\phi \otimes X\, , \,\psi  \otimes Y\,]^{\,\widehat{\,}}
  &=&  (\phi\cdot\psi)\otimes\,(XY) - (\psi \cdot \phi) \otimes (YX)+ (X|Y)\, c(\phi , \psi )\,a \,,
\\[0,3cm] 
 [\,a\,, \,\phi\otimes X\,] ^{\,\widehat{\,}}&=&0\, ,
  \end{eqnarray}
for \(X,Y\in U(\mathfrak{g})\) and \(\phi,\,\psi\,\in S^3\mathbf{H} \, \). 
\end{theorem}

\medskip

As a Lie subalgebra of \(S^3\mathfrak{g}^{\mathbf{H}}\) we have \(\mathbf{C}[\phi^{\pm}]\otimes U(\mathfrak{g})\).   
\begin{definition}~~
We denote by \(\widehat{\mathfrak{g}}(a)\) the extension of the Lie algebra \(\mathbf{C}[\phi^{\pm}]\otimes U(\mathfrak{g})\) by the 1-dimensional center \(\mathbf{C}a\) associated to the cocycle \(c\):
\begin{equation}
\widehat{\mathfrak{g}}(a)=\mathbf{C}[\phi^{\pm}]\otimes U(\mathfrak{g})\,\oplus (\mathbf{C}a).
\end{equation}
The Lie bracket is given by
 \begin{eqnarray}
 [\,\phi \otimes X\, , \,\psi  \otimes Y\,]^{\,\widehat{\,}}
  &=&  (\phi\cdot\psi)\otimes\,(XY) - (\psi \cdot \phi) \otimes (YX)+ (X|Y)\, c(\phi , \psi )\,a \,,
\\[0,3cm] 
 [\,a\,, \,\phi\otimes X\,] ^{\,\widehat{\,}}&=&0\, ,
  \end{eqnarray}
for \(X,Y\in U(\mathfrak{g})\) and \(\phi,\,\psi\,\in \mathbf{C}[\phi^{\pm}] \, \). 
\end{definition}

\subsection{Extension of \(\widehat{\mathfrak{g}}(a)\)  by the derivation}
We introduced the radial derivative \(d_0\) acting on \(S^3\mathbf{H}\).   \(d_0\) preserves the space of spinors of Laurent polynomial type \(\mathbf{C}[\phi^{\pm}]\).   
The derivation \(d_0\) on \(\mathbf{C}[\phi^{\pm}]\) is extended to a derivation of the Lie algebra \(\,\mathbf{C}[\phi^{\pm}]\otimes U(\mathfrak{g})\,\) by 
\begin{equation}
   \,d_0\,(\phi \otimes X\,)\,=\,(d_0\phi \,)\otimes X\,.
\end{equation}
In fact we have from (\ref{leibnitz1})
  \begin{eqnarray*}
 && d_0\left(\,[\,\phi_1\otimes X_1\,,\,\phi_2\otimes X_2\,]^{\,\widehat{\,}}\,\,\right)=
  d_0\left(\,(\phi_1\phi_2)\,\otimes(X_1X_2)\,-\,(\phi_2\phi_1)\otimes (X_2X_1)\,\right)\\[0.2cm]
  &&\,=\,(d_0\phi_1 \cdot\phi_2)\otimes(X_1X_2)\,-\,(\phi_2\cdot d_0\phi_1)\otimes (X_2X_1)
 +\,(\phi_1 \cdot d_0\phi_2)\otimes (X_1X_2)\,  \\[0.2cm]
  && \qquad -\,(d_0\phi_2\cdot \phi_1)\otimes (X_2X_1)\,.\end{eqnarray*}
   On the other hand 
\begin{eqnarray*}
&&\,[\,d_0(\phi_1\otimes X_1)\,,\,\phi_2\otimes X_2\,]^{\,\widehat{\,}}\,+\,
[\,\phi_1\otimes X_1\,,\,d_0(\phi_2\otimes X_2)\,]^{\,\widehat{\,}}
  \\[0.2cm]
 && \,=\,(d_0\phi_1 \cdot\phi_2)\otimes(X_1X_2)\,-\,(\phi_2\cdot d_0\phi_1)\otimes (X_2X_1)
\,
 +\,(\phi_1 \cdot d_0\phi_2)\otimes (X_1X_2)\,\\[0.2cm]
&&
 \qquad -\,(d_0\phi_2\cdot \phi_1)\otimes (X_2X_1)\,   +\,(X_1\vert X_2)\left(\,c(d_0\phi_1,\,\phi_2)\,+\,c(\phi_1,\,d_0\phi_2)\right)a\,.
\end{eqnarray*}  
 Since \(c(d_0\phi_1,\,\phi_2)\,+\,c(\phi_1,\,d_0\phi_2)=0\) from Proposition \ref{deriv} we have 
 \begin{eqnarray*}
 && d_0\left(\,[\,\phi_1\otimes X_1\,,\,\phi_2\otimes X_2\,]^{\,\widehat{\,}}\,\,\right)\,\\[0.2cm]
\qquad && =\,
 [\,d_0(\phi_1\otimes X_1)\,,\,\phi_2\otimes X_2\,]^{\,\widehat{\,}}\,+\,
[\,\phi_1\otimes X_1\,,\,d_0(\phi_2\otimes X_2)\,]^{\,\widehat{\,}}\,\,.
\end{eqnarray*}
Thus \(d_0\) is a derivation that acts on the Lie algebra \(\mathbf{C}[\phi^{\pm}]\otimes U(\mathfrak{g})\,\).  

  We denote by \(\widehat{\mathfrak{g}}\,\) the Lie algebra that is obtained by adjoining a derivation \(d\) to  \(\widehat{\mathfrak{g}}(a)\) which acts on  \(\mathbf{C}[\phi^{\pm}]\otimes U(\mathfrak{g})\,\) as \(d_0\) and which kills \(a\).    More explicitly we have the following
\begin{theorem}
 Let \(a\) and  \(\,d\) be indefinite elements. We consider the \(\mathbf{C}\) vector space:
\begin{equation}
\widehat{\mathfrak{g}}\,=\, \left(\mathbf{C}[\phi^{\pm}]\otimes U(\mathfrak{g})\right) \oplus( \mathbf{C}\, a )\oplus (\mathbf{C}d)\,,
\end{equation}
 and define the following bracket on $\widehat{\mathfrak{g}}\).   
  For \(X,Y\in U(\mathfrak{g})\) and \(\phi,\,\psi\,\in \,\mathbf{C}[\phi^{\pm}]\, \), we put 
 \begin{eqnarray}
 [\,\phi \otimes X\, , \,\psi  \otimes Y\,]_{\widehat{\mathfrak{g}}} &=&
  [\,\phi \otimes X\, , \,\psi  \otimes Y\,]^{\,\widehat{\,}}    \label{brac2}  \\[0.2cm]
  &=&  (\phi\cdot\psi)\otimes\,(XY) - (\psi \cdot \phi) \otimes (YX)+ (X|Y)\, c(\phi , \psi )\,a \, , \nonumber
\\[0,3cm] 
 [\,a\,, \,\phi\otimes X\,] _{\widehat{\mathfrak{g}}}&=&0\,, \qquad
  [\,d, \phi \otimes X\,] _{\widehat{\mathfrak{g}}}=\,d_0\phi \otimes X\, , \\[0,3cm] 
  [\,d\,,\,a\,]_{\widehat{\mathfrak{g}}}\,&=&0\,.
  \end{eqnarray}
      Then \( \left(\, \widehat{\mathfrak{g}} \, , \, [\,\cdot,\cdot\,]_{ \widehat{\mathfrak{g} } } \,\right) \) becomes a Lie algebra.
\end{theorem}

{\it Proof}

It is enough to prove the following Jacobi identity:
 \begin{equation*}
 [\,[\,d\,, \,\phi_1   \otimes X_1 \,]_{\widehat{\mathfrak{g}}}\,,\, \phi_2 \otimes X_2\,]_{\widehat{\mathfrak{g}}}
+[\,[\phi_1  \otimes X_1 , \phi_2 \otimes X_2\,]_{\widehat{\mathfrak{g}}}\,,\,d\,]_{\widehat{\mathfrak{g}}}
\,+\,[\,[\phi_2 \otimes X_2 , \,d\,]_{\widehat{\mathfrak{g}}} ,\, \phi _1   \otimes X_1\,]_{\widehat{\mathfrak{g}}}=0.
\end{equation*}
In the following we shall abbreviate the bracket \([\,,\,]_{\widehat{\mathfrak{g}}}\,\) simply to \([\,\,,\,\,]\).    
  We have
\begin{align*}
[\,[\,d\,, \,\phi_1  \otimes X_1 \,]\,, \phi_2 \otimes X_2\,] =&
[\,d_0\phi_1\otimes X_1,\,\phi_2 \otimes X_2\,] \\[0.2cm]
=& \,( \,d_0\phi_1\cdot \phi_2 )\otimes\,\,(X_1X_2 )
-\,(\,\phi_2\cdot d_0\phi_1\,)\otimes (X_2X_1)\\[0.2cm]
&+
(X_1\vert X_2)c(\,d_0\phi_1\,,\,\phi_2)\, a\,.
\end{align*}
Similarly
\begin{align*}
[\,[ \,\phi_2  \otimes X_2, \,d \,] , \phi_1 \otimes X_1\,]=&
(\phi_1\cdot d_0\phi_2) \otimes (X_1X_2 ) 
-\,(\,d_0\phi_2 \cdot\phi_1)\otimes (X_2X_1)\\[0.2cm]
&+
(X_1\vert X_2)\,c(\phi_1,\,d_0\phi_2\,)\, a\,.  \\[0.2cm]
[\,[\phi_1  \otimes X_1 , \phi_2  \otimes X_2\,]\,,\,d\,]=&
-\bigl[\,d\,,\,
(\phi_1\cdot \phi_2 )\otimes (X_1X_2 )-(\phi_2\cdot\phi_1)\otimes (X_2X_1)+
(X_1\vert X_2)c (\phi_1\,, \phi_2)\,a\,\bigr]\\[0.2cm]
=&-\,d_0(\phi_1\cdot\phi_2) \otimes\,(X_1X_2 )
+ d_0(\phi_2\cdot\phi_1)\otimes (X_2X_1) \,.
\end{align*}
  The sum of three equations  vanishes by virtue of (\ref{leibnitz1}) and Proposition \ref{deriv}.   
 \hfill\qed
 \par\medskip
 
 Remember from Definition \ref{homog} that  \(\mathbf{C}[\phi^{\pm};\,N]\) denotes the subspace in \(\mathbf{C}[\phi^{\pm}]\) generated by the products 
  \(\phi_1\cdots\phi_n\,\) with each \(\phi_i\) being \(\phi_i=\phi^{+(m_i,l_i,k_i)}\) or  \(\,\phi_i=\phi^{-(m_i,l_i,k_i)}\), \(i=1,\cdots,n\), such that 
  \[\sum_{i;\,\phi_i=\phi^{+(m_i,l_i,k_i)}}m_i\,-\,\sum_{i;\,\phi_i=\phi^{-(m_i,l_i,k_i)}}(m_i+3)\,=\,N\,.\]
   \begin{proposition}
  The centralizer of \(\,d\) in \(\,\widehat{\mathfrak{g}}\,\) is given by
\begin{equation}
(\,\mathbf{C}[\phi^{\pm};\,0]\,\otimes U(\mathfrak{g})\,)\,\oplus \, \mathbf{C}a\,\oplus\mathbf{C}d\,.
 \end{equation}
   \end{proposition}  
The proposition follows from (\ref{d0value}) .

\section{Structure of  \(\widehat{\mathfrak{g}}\)}

\subsection{The weight space decomposition of \(U(\mathfrak{g})\)}

Let \((\,\mathfrak{g}\,,\,[\,,\,]_{\mathfrak{g}}\,) \) be a simple Lie algebra.      Let 
\(\mathfrak{h}\) be a Cartan subalgebra of  \(\mathfrak{g}\) and \(\mathfrak{g}= \mathfrak{h}\oplus \sum_{\alpha\in \Delta}\,\mathfrak{g}_{\alpha}\) be the root space decomposition with the root space  
\(\,\mathfrak{g}_{\alpha}=\{X\in\mathfrak{g};\,ad(h)X\,=\,<\alpha,h>X, \quad\forall h\in \mathfrak{h}\}\) .      Here \(\Delta=\Delta(\mathfrak{g},\mathfrak{h})\) is the set of roots and \(\dim\,\mathfrak{g}_{\alpha}=1\).    
     Let \(\Pi=\{\alpha_i;\,i=1,\cdots,r={\rm rank}\,\mathfrak{g}\}\subset \mathfrak{h}^{\ast}\) be the set of simple roots and  \(\{\alpha_i^{\vee}\,;\,i=1,\cdots,r\,\}\subset \mathfrak{h}\) be the set of simple coroots.   The Cartan matrix \(A=(\,a_{ij}\,)_{i,j=1,\cdots,r}\) is given by \(a_{ij}=\left\langle \alpha_i^{\vee},\,\alpha_j \right\rangle\).      Fix a standard set of generators \(\,H_i=\alpha_i^{\vee}\), \(X_i\in \mathfrak{g}_{\alpha_i}\),  \(Y_i\in \mathfrak{g}_{-\alpha_i}\), so that \([\,X_i,\,Y_j ]=H_j\delta_{ij}\), \(\,[H_i,\,X_j]=a_{ji}X_j\) and \(\,[H_i,\,Y_j]=-a_{ji}Y_j\).     Let 
\(\Delta_{\pm}\) be the set of positive ( respectively negative )  roots of \(\mathfrak{g}\) and put 
\[\mathfrak{n}_{\pm}=\sum_{\alpha \in \Delta_{\pm}}\,\mathfrak{g}_{\alpha}\,.\]
Then \(\mathfrak{g}= \mathfrak{n}_+ \oplus \mathfrak{h} \oplus \mathfrak{n}_-\).    The enveloping algebra  
 \(U(\mathfrak{g})\) of \(\mathfrak{g}\) has the direct sum decomposition: 
\begin{equation}
U(\mathfrak{g})=U( \mathfrak{n}_-)\cdot U(\mathfrak{h})\cdot U(\mathfrak{n}_+)\,.
\end{equation}
  In the following we summarize the known results on the representation \((ad(\mathfrak{h}),\,U(\mathfrak{g})\,)\), \cite{D, Ma}.  
The set 
\[ \{\,Y_1^{m_1}\cdots\,Y_r^{m_r}H_1^{l_1}\cdots H_r^{l_r}X_1^{n_1}\cdots X_r^{n_r}\,
;\quad m_i,\,n_i,\, l_i\in \mathbf{N}\cup{0}\, \}.\]
 forms a basis of the enveloping algebra \(U(\mathfrak{g})\).    
 The  adjoint action of \( \mathfrak{h}\) is extended to that on \(U(\mathfrak{g})\):
 \[ad(h)(x\cdot y)=\,(\,ad(h)x\,)\cdot y\,+\,x\cdot(\,ad(h)y\,)\,.\]
         \( \lambda\in \mathfrak{h}^{\ast}\) is called a weight of the representation \((U(\mathfrak{g}),\,ad(\mathfrak{h})\,)\)  if there exists a non-zero \(x\in U(\mathfrak{g})\) such that  \(ad(h)x=hx-xh=\lambda(h) x\) for all \(h\in\mathfrak{h}\) .   Let \(\Sigma\) be the set of weights of the representation \((U(\mathfrak{g}),\,ad(\mathfrak{h})\,)\).   The weight space for the weight \(\lambda\)  is by definition 
 \[\mathfrak{g}^U_\lambda\,=\,\{x\in U(\mathfrak{g})\,;\quad ad(h)x=\lambda(h)x,\quad\forall h\in\mathfrak{h}\}.\]
Let  \(\lambda=\sum_{i=1}^r\,n_i\alpha_{i}-\sum_{i=1}^r\,m_i\alpha_{i}\,\), \(\,n_i,\,m_i\geq 0\).      For any \(\,l_1,\,l_2,\cdots,l_r\geq 0\,\),  
\[ \overline X_{\lambda}=\,Y_{1}^{m_1}\cdots\,Y_{r}^{m_r}H_1^{l_1}\cdots H_r^{l_r}X_{1}^{n_1}\cdots X_{r}^{n_r}\,\,\]
gives a weight vector with the weight \(\lambda\,\);  \(\,\overline X_\lambda\in \mathfrak{g}^U_\lambda\,\).
  Conversely any weight \(\lambda\) may be written in the form \(\lambda=\sum_{i=1}^r\,n_i\alpha_{i}-\sum_{i=1}^r\,m_i\alpha_{i}\,\), though the coefficients  \(n_i\,, m_i\) are not uniquely determined.    
\begin{lemma}
\begin{enumerate}
\item
 The set of weights of the adjoint representation \(\left(U(\mathfrak{g}),\,ad(\mathfrak{h})\right)\)    
is 
 \begin{equation}
 \Sigma=\{\,\sum\,k_i\alpha_i\, ;\quad \alpha_i\in \Pi\,,\,k_i\in \mathbf{Z}\,\}.
 \end{equation} 
If we denote 
  \begin{equation}
 \Sigma_{\pm}=\{\,\pm \sum\,n_i\alpha_i\,\in\Sigma\, ;\quad n_i>0\,\} 
\end{equation}
 then \(\Sigma_{\pm}\cap\Delta=\Delta_{\pm}\).
 \item
  If \(\lambda\in \Sigma\) then \(-\lambda\in\Sigma\).
\item
For each \(\, \lambda=\,\sum_{i=1}^r\,k_i\alpha_i\, \in \Sigma\,\), 
 \(\,\mathfrak{g}^U_{\lambda}\)  is generated by the basis 
\[\overline X_{\lambda}(l_1,\cdots,l_r,m_1,\cdots,m_r,n_1,\cdots,n_r)=\,Y_{1}^{m_1}\cdots\,Y_{r}^{m_r}H_1^{l_1}\cdots H_r^{l_r}X_{1}^{n_1}\cdots X_{r}^{n_r}\,\,\] 
with \(\,n_i, \,m_i,\, l_i\in \mathbf{N}\cup 0\,\) such that \(\,k_i=n_i-m_i\,\), \(\,i=1,\cdots,r\,\).\\
In particular 
 \(\,\mathfrak{g}^U_{0}\)  is generated by the basis  
\[\overline X_0
(l_1,\cdots,l_r,n_1,\cdots,n_r,n_1,\cdots,n_r)=\,Y_{1}^{n_1}\cdots\,Y_{r}^{n_r}\,H_1^{l_1}\cdots H_r^{l_r}\,X_{1}^{n_1}\cdots X_{r}^{n_r}\,\] 
with \(\,n_i,\, l_i\in \mathbf{N}\cup 0\,\) , \(\,i=1,\cdots,r\,\).    In particular 
 \[\,U(\mathfrak{h})\subset \mathfrak{g}^U_0\,.\]
\item
\begin{equation}
[\,\mathfrak{g}^U_\lambda\,,\,\mathfrak{g}^U_\mu\,]\,\subset \,\mathfrak{g}^U_{\lambda+\mu}\,,
\end{equation}
\end{enumerate}
\end{lemma} 

\subsection{ Weight space decomposition of \(\,\widehat{\mathfrak{g}}\,\)}

In the following we shall investigate the Lie algebra structure of 
\begin{equation}
\widehat{\mathfrak{g}}\,=\,
 \left(\,\mathbf{C}[\phi^{\pm}] \otimes U(\mathfrak{g})\,\right) \oplus (\mathbf{C}\,a )\oplus  ( \mathbf{C}d)\,.
\end{equation}
Remember that the Lie bracket was defined by 
 \begin{eqnarray*}
 [\,\phi \otimes X\, , \,\psi  \otimes Y\,]_{\widehat{\mathfrak{g}}} &=& (\phi\,\psi)\otimes (XY)-(\psi\phi)\otimes (YX) +\, (X|Y)\, c(\phi , \psi )\,a\,,  
\\[0,3cm] 
 [\,a\,, \phi\otimes X\,] _{\widehat{\mathfrak{g}}}=0\,, &\quad &\, [\,a,\,d\,]_{\widehat{\mathfrak{g}}} =0\,, \\[0.2cm]
  [\,d, \phi \otimes X\,] _{\widehat{\mathfrak{g}}}&=&d_0 \phi \otimes X\, ,
  \end{eqnarray*}
 for \(\,X,Y\in U(\mathfrak{g})\).     
 Since 
\(
\phi ^{ +(0,0,1)} = \begin{pmatrix} 1 \\ 0 \end{pmatrix}
\) we identify  \(X\in U(\mathfrak{g})\) with $\phi ^{+ (0,0,1)} \otimes X$.    Thus we look  
$\mathfrak{g}$  as a Lie subalgebra of $\,\widehat{\mathfrak{g}}\,$: 
\begin{equation}
\left[\phi^{+(0,0,1)}\otimes X,\,\phi^{+(0,0,1)}\otimes Y\right]_{\widehat{\mathfrak{g}}} =\left[X,Y\right]_{\mathfrak{g}} \,,
\end{equation}
and we shall write \(\phi ^{+ (0,0,1)} \otimes X\) simply as \(X\).

  Let 
\begin{equation}
\widehat{\mathfrak{h}}\,=\,
 (\,(\, \mathbf{C}\,\phi ^{+(0, 0,1)}\,) \otimes\mathfrak{h} )\,\oplus (\mathbf{C}\,a )\oplus (\mathbf{C}\,d )\,
 =
 \mathfrak{h}\oplus (\mathbf{C}\,a )\oplus (\mathbf{C} d)\,.
\end{equation}

 We write \(\hat h=h+sa+t d\in \widehat{\mathfrak{h}}\) with \(h\in \mathfrak{h}\) and \(s,\,t\in\mathbf{C}\).     
For any \(h\in \mathfrak{h}\), \(\phi\in \mathbf{C}[\phi ^{\pm}] \) and \(X\in U(\mathfrak{g})\),  it holds that 
\begin{eqnarray*}
[\,\phi^{+(0,0,1)}\otimes h,\,\phi\otimes X\,]_{\widehat{\mathfrak{g}}}&=&\phi\otimes (\,hX-Xh)\,, \label{adH1}\\[0.2cm]
[\,d,\,\phi\otimes X\,]_{\widehat{\mathfrak{g}}}&=&(d_0\phi\,)\otimes X\,, \label{adH1}\\[0.2cm]
\quad [\,\phi^{+(0,0,1)}\otimes h, \,a\,]_{\widehat{\mathfrak{g}}}&=&0\,,\quad 
[\,\phi^{+(0,0,1)}\otimes h,\,d\,]_{\widehat{\mathfrak{g}}}=0,\quad [\,d,a\,]_{\widehat{\mathfrak{g}}}=0\,.
 \end{eqnarray*}
 Then the adjoint actions of \(\hat h= h+sa+t d \in\widehat{\mathfrak{h}}\) on \(\widehat{\mathfrak{g}}\) is written as follows.    
\begin{equation}
ad(\hat h)\,(\phi\otimes X+\mu a+\nu d)=\phi\otimes ( hX-Xh) +  t 
d_0\phi\otimes X\,,\end{equation}
for \(\xi=\phi\otimes X+\mu a+\nu d\in \widehat{\mathfrak{g}}\). 

An element $\lambda$ of the dual space $\mathfrak{h}^* $ of $\mathfrak{h} $ can be regarded as an element
of $\widehat{\mathfrak{h}}^{\,\ast}$ by putting
\begin{align}
\left\langle \lambda , a \right\rangle= 
\left\langle \lambda , d \right\rangle = 0.
\end{align}
So  $\Delta \subset \mathfrak{h}^*$ is seen to be a subset of $\widehat{\mathfrak{h}}^{\,*}$.    
We define the elements $\delta\,,\,\Lambda_0 \in \widehat{\mathfrak{h}}^{\,*}$  by
\begin{align}
\left\langle\delta , \alpha _i ^{\vee} \,\right\rangle &= \left\langle\,\Lambda_0 , \alpha _i ^{\vee} \,\right\rangle = 0,  \qquad (1 \leqq i \leqq  r),\\[0.2cm]
 \left\langle\,\delta , a\right\rangle &= 0\,,  \qquad \left\langle\,\delta , d\,\right\rangle = 1,\\[0.2cm]
 \left\langle\,\Lambda_0 , a\right\rangle &= 1\,,  \qquad \left\langle\,\Lambda_0 , d\,\right\rangle = 0.
\end{align}  
Then  the set  \(\{\,\alpha _1,\cdots,\alpha_r , \,\Lambda_0 ,\,\delta \,\} \)  forms a basis of  $\widehat{\mathfrak{h}}^{\,*}$.     Similarly \(\Sigma\) is a subset of \(\widehat{\mathfrak{h}}^{\ast}\).

  Since \(\widehat{\mathfrak{h}}\) is a commutative subalgebra of \(\widehat{\mathfrak{g}}\,\), 
 \(\,\widehat{\mathfrak{g}}\) is decomposed into a direct sum of the simultaneous eigenspaces of \(ad\,(\hat h)\), \(\,\hat h\in \widehat{\mathfrak{h}}\,\).      
 
 For \(\lambda=\gamma+k_0\delta\in \widehat{\mathfrak{h}} ^{\,\ast}\), \(\gamma=\sum_{i=1}^r\,k_i\alpha_i\in\Sigma\),  \(k_i\in \mathbf{Z},\,i=0,1,\cdots,r\), we put,
\begin{equation}
\widehat{\mathfrak{g}}_{\lambda}=\left\{\xi\in \widehat{\mathfrak{g}}\,;\quad \, [\,\hat h,\,\xi\,]\,=\,\langle \lambda, \hat h\rangle\,\xi\quad\mbox{ for }\, \forall\hat h\in \widehat{\mathfrak{h}}\,\right\}.\end{equation}
\(\lambda\)  is called a weight of \(\,\widehat{\mathfrak{g}}\,\)  if \(\, \widehat{\mathfrak{g}}_{\lambda}\neq 0\).     \(\,\widehat{\mathfrak{g}}_{\lambda}\) is called the weight space of  \(\lambda\,\).   

Let  \(\widehat{\Sigma}\) denote the set of weights of  the representation \(\left(\widehat{\mathfrak{g}}\,,ad(\widehat{\mathfrak{h}})\right)\).  

\begin{theorem}   
\begin{enumerate}
\item
\begin{eqnarray*}
\widehat{\Sigma} &=& \left\{ \frac{m}{2} \delta + \lambda;\quad \lambda \in \Sigma\,,\,m\in\mathbf{Z}\,\right\} \\[0.2cm]
&& \bigcup \left\{ \frac{m}{2} \delta ;\quad  m\in \mathbf{Z} \, \right\}  \,.
\end{eqnarray*}
\item
For \(\lambda\in \Sigma\), \(\lambda\neq 0\) and \(m\in \mathbf{Z}\), we have 
 \begin{equation}
\widehat{ \mathfrak{g}}_{\frac{m}{2}\delta+ \lambda}\,=\mathbf{C}[\phi ^{\pm};\,m\,] \otimes \mathfrak{g} _{ \lambda}^U\,.
 \end{equation}
 \item
 \begin{eqnarray*}  
  \widehat{ \mathfrak{g}}_{0\delta}&=& (\,\mathbf{C}[\phi^{\pm};0\,]  \otimes \mathfrak{g}^U_0\,)\oplus(\mathbf{C}a)\oplus(\mathbf{C}d)\,\supset\,\widehat{\mathfrak{h}}\,, 
\\[0.2cm]
 \widehat{ \mathfrak{g}}_{\frac{m}{2}\delta}&=&  \,\mathbf{C}[\phi^{\pm};\,m\,]  \otimes \mathfrak{g}^U_0\,\,, \quad\mbox{for  \(0\neq  m\in\mathbf{Z} \) . }\,
 \end{eqnarray*}
  \item
 \(\widehat{ \mathfrak{g}}\) has the following decomposition:
\begin{equation}
\widehat{ \mathfrak{g}}\,=\, \bigoplus_{m\in\mathbf{Z} }\, \widehat{ \mathfrak{g}}_{\frac{m}{2}\delta}\,\bigoplus\,\,\bigoplus_{\lambda\in \Sigma,\, m\in\mathbf{Z} }\, 
\widehat{ \mathfrak{g}}_{\frac{m}{2}\delta+\lambda}\,
\end{equation}
\end{enumerate}
\end{theorem}

{\it Proof} 

First we prove the second assertion.      
Let \(X\in\mathfrak{g}_{\lambda}^U\) for a \(\lambda\in \Sigma\), \(\lambda\neq 0\), and let  \(\varphi\in \mathbf{C}[\phi^{\pm};\,m]\)  for a \(m\in \mathbf{Z}\).     We have, for any \(h\in\mathfrak{h}\), 
\begin{eqnarray*}
[\,\phi ^{+ (0,0,1)} \otimes h , \,\varphi \otimes X\,] _{ \widehat{\mathfrak{g}}}&=&
 \varphi \otimes (hX-Xh) \,
= \left\langle \lambda , h\right\rangle \varphi \otimes X,
\\[0.2cm]
[\,d, \,\varphi \otimes X\,]_{ \widehat{\mathfrak{g}}}&=& \frac{m}{2} \varphi\otimes X,
\end{eqnarray*}
that is, for every \(\hat{h} \in \widehat{\mathfrak{h}}\), we have 
\begin{equation}
[\,\hat h\,, \varphi \otimes X]_{ \widehat{\mathfrak{g}}} = \left\langle \frac{m}{2} \delta + \lambda \,, \,\hat h\,\right\rangle (\varphi\otimes X)\,.
\end{equation}
Therefore we have \(\varphi\otimes X\in \widehat{\mathfrak{g}}_{\frac{m}{2}\delta+ \lambda}\). 

Conversely,  for a given \(m\in\mathbf{Z}\)  and  a \(\xi\in \widehat{\mathfrak{g}}_{\frac{m}{2}\delta+ \lambda}\), we shall show that \(\xi\)  has the form  \(\,\phi\otimes X\,\) with \(\phi\in \mathbf{C}[\phi^{\pm};m]\,\) and \(X\in \mathfrak{g}^U_ \lambda\,\) .   Let 
\(\xi=\phi\otimes X+\mu a+\nu d\) for \(\phi\in \mathbf{C}[\phi^{\pm}]\), \(\,X\in U(\mathfrak{g})\) and \(\mu,\,\nu\in\mathbf{C}\).     \(\phi\) is  decomposed to the sum  
\[\phi=\sum_{n\in \mathbf{Z}}\,\phi_n\]
by the homogeneous degree; \(\phi_n\in \mathbf{C}[\phi^{\pm};n]\).
 We have 
\begin{eqnarray*}
&&[\hat h,\xi]=[\,\phi^{+(0,0,1)}\otimes h+ sa+td\,, \,\phi \otimes X\,+\mu a+\nu d\,]
 =\,\phi\otimes [\,h\,,\,X\,]\\[0.2cm]
 &&\qquad + \,t (\,\sum_{n\in \mathbf{Z}} \,\frac{n}{2} \phi_n  \,\otimes X\,)
\end{eqnarray*}
for any \(\hat h=\phi^{+(0,0,1)}\otimes h+sa+td\in \widehat{\mathfrak{h}}\).
From the assumption we have 
\begin{eqnarray*}
 [\,\hat h,\xi\,]\,&=&\,\langle\, \frac{m}{2}\delta+ \lambda\,,\,\hat h\,\rangle\, \xi\,\\[0.2cm]
&=& < \lambda,h>\phi\otimes X\, +(\frac{m}{2}t+< \lambda,h>)(\mu a+\nu d)\,\\[0.2cm]
&& \quad+\,
\frac{m}{2}t\,
  (\sum_{n} \,\phi_{n} )\otimes X .
\end{eqnarray*}
Comparing the above two equations we have \(\mu=\nu=0\), and $\phi_n = 0$ for all \(n\) 
except for $n = m$.      Therefore    $ \phi \in\mathbf{C}[\phi^{\pm};m]$.   We also have  $[\hat h , \xi] =\phi\otimes [h,X]= \langle  \lambda ,\,h\rangle \,\phi \otimes X$ for all \(\hat h=\phi^{+(0,0,1)}\otimes h+sa+td \in\widehat{\mathfrak{h}}\).     Hence  $X \in \mathfrak{g}^U_{ \lambda}$ and 
  \(\xi= \phi_m\otimes X \in \widehat{\mathfrak{g}}_{\frac{m}{2}\delta + \lambda}\,\).   We have proved
  \begin{equation*}
  \widehat{\mathfrak{g}}_{\frac{m}{2}\delta+ \lambda}=\mathbf{C}[\phi^{\pm};m] \otimes \mathfrak{g}^U_{ \lambda}\,.
 \end{equation*}
 The proof of the third assertion is also carried out by the same argument as above if  we revise it for the case \(\lambda=0\) .      The above discussion yields the first and the fourth assertions.

  \hfill\qed
  \medskip
  
\vspace{0.5cm}

\begin{proposition}
We have the following relations: 
\begin{enumerate}
\item
\begin{equation}
\left[\, \widehat{\mathfrak{g}}_{\frac{m}{2}\delta+\alpha}\,,\, \widehat{\mathfrak{g}}_{\frac{n}{2}\delta+\beta}\,\right ]_{ \widehat{\mathfrak{g}}}\,\subset \,
 \widehat{\mathfrak{g}}_{\frac{m+n}{2}\delta+\alpha+\beta}\,\,,
 \end{equation}
 for \(\alpha,\,\beta \in \widehat{\Sigma}\) and for  \(m,n\in\mathbf{Z}\).
\item
\begin{equation}
\left[\, \widehat{\mathfrak{g}}_{\frac{m}{2}\delta}\,,\, \widehat{\mathfrak{g}}_{\frac{n}{2}\delta}\,\right ]_{ \widehat{\mathfrak{g}}}\,\subset \,
 \widehat{\mathfrak{g}}_{\frac{m+n}{2}\delta}\,, 
 \end{equation}
  for  \(m,n\in\mathbf{Z}\).
\end{enumerate}
\end{proposition}

{\it Proof}

Let $\phi \otimes X \in \widehat{\mathfrak{g}}_{\frac{m}{2}\delta+\alpha}$ and
$\psi \otimes Y \in \widehat{\mathfrak{g}}_{\frac{n}{2}\delta+\beta}$.   Then we have, for \(h\in\mathfrak{h}\),  
\begin{align*}
[\,h, [\,\phi \otimes X , \psi \otimes Y]\,] &= 
-[\,\phi \otimes X , [\,\psi \otimes Y , h\,]\,] - [\,\psi \otimes Y , [\,h , \phi \otimes X]\,] \\
&= <\beta, h> [\,\phi \otimes X, \psi \otimes Y\,] + <\alpha, h>[\,\phi \otimes X, \psi \otimes Y\,] \\
&= <\alpha + \beta , h> [\,\phi \otimes X, \psi \otimes Y\,].
\end{align*}
On the other hand,
\begin{align*}
[\,d, [\,\phi \otimes X, \psi \otimes Y\,]\,]
&= -[\,\phi \otimes X , [\psi \otimes Y , d]\,] - [\,\psi \otimes Y , [\,d , \phi \otimes X]\,] \\
&= \frac{m + n}{2} [\,\phi \otimes X , \psi \otimes Y\,]\,. 
\end{align*}
Hence
\begin{equation}
[\,\widehat h, [\,\phi \otimes X , \psi \otimes Y\,]\,] = \left<\frac{m+n}{2}\delta  + \alpha + \beta \,,\, \widehat h \right> [\,\phi \otimes X , \psi \otimes Y\,]
\end{equation}
for any \(\widehat h\in \widehat{\mathfrak{h}}\).    Therefore 
\begin{equation}
\left[\, \widehat{\mathfrak{g}}_{\frac{m}{2}\delta+\alpha}\,,\, \widehat{\mathfrak{g}}_{\frac{n}{2}\delta+\beta}\,\right ]_{ \widehat{\mathfrak{g}}}\,\subset \,
 \widehat{\mathfrak{g}}_{\frac{m+n}{2}\delta+\alpha+\beta}\,\,,
 \end{equation}
 The same calculation for \(\phi\otimes H\in \widehat{\mathfrak{g}}_{\frac{m}{2}\delta}\) and  \(\psi\otimes H^{\prime}\in \widehat{\mathfrak{g}}_{\frac{n}{2}\delta}\,\) yields 
 \begin{equation}
\left[\, \widehat{\mathfrak{g}}_{\frac{m}{2}\delta}\,,\, \widehat{\mathfrak{g}}_{\frac{n}{2}\delta}\,\right ]_{ \widehat{\mathfrak{g}}}\,\subset \,
 \widehat{\mathfrak{g}}_{\frac{m+n}{2}\delta}\,.
\end{equation}
 \hfill\qed

 \subsection{ generators of \(\widehat{\mathfrak{g}}\)}

Let \(\{\alpha_i\}_{i=1,\cdots,r}\subset \mathfrak{h}^{\ast}\) be the set of simple roots and  \(\{h_i\}_{i=1,\cdots,r}\subset \mathfrak{h}\) be the set of simple coroots.   \(e_i\), \(f_i\), \(i=1,\cdots,r\), denote the Chevalley generators;
\begin{eqnarray*} 
[\,e_i,\,f_j\,]&=& \delta_{ij}h_i,\\[0.2cm]
[\,h,\,e_i\,] &=& \alpha_i(h)\,,\quad [\,h,\,f_i \,]\,=-\alpha_i(h),\quad \mbox{for } \forall h\in\mathfrak{h}.
\end{eqnarray*}
Let \(A=(\,a_{ij}\,)_{i,j=1,\cdots,r}\) be the Cartan matrix of \(\mathfrak{g}\); \(a_{ij}=\alpha_i(h_j)\).
  
By the natural embedding of \(\mathfrak{g}\) in \(\widehat{\mathfrak{g}}\) we have the vectors 
\begin{eqnarray}
\widehat h_i&=&\phi^{+(0,0,1)}\otimes h_i\,\in \widehat{\mathfrak{h}},\,\\[0.2cm]
\widehat e_i&=&\phi^{+(0,0,1)}\otimes e_i\,\in \widehat{\mathfrak{g}}_{0\delta+\alpha_i},\quad \widehat f_i=\phi^{+(0,0,1)}\otimes f_i\,\in \widehat{\mathfrak{g}}_{0\delta-\alpha_i},\qquad i=1,\cdots,r\,.
\end{eqnarray}
It is easy to verify the relations:
\begin{eqnarray}
\left[\widehat e_i\,,\widehat f_j\,\right]_{ \widehat{\mathfrak{g}}} &=&\,\delta_{ij}\,\widehat h_i\,,\\[0.2cm]
\left[\widehat h_i\,,\widehat e_j\,\right ]_{ \widehat{\mathfrak{g}}}&=&\,a_{ij}\,\widehat e_j,\quad 
\left[\widehat h_i\,,\widehat f_j\,\right]_{ \widehat{\mathfrak{g}}} =\,- a_{ij}\,\widehat f_j\,,\quad 1\leq i,j\leq r .
\end{eqnarray}
We have obtained a part of generators of \(\widehat{\mathfrak{g}}\) that come naturally from \(\mathfrak{g}\).

We recall that for an affine Lie algebra \((\mathbf{C}[t,t^{-1}]\otimes\mathfrak{g})\oplus(\mathbf {C}a)\oplus(\mathbf{C}d)\) there is a special Chevalley generator coming from the irreducible representation spaces \(t^{\pm1}\otimes \mathfrak{g}\) of the simple Lie algebra  \(\mathfrak{g}\).    Let  \(\theta\) be the highest root of \(\mathfrak{g}\) and suppose  that \(e_\theta\in\mathfrak{g}_\theta\)  and 
\(e_{- \theta}\in \mathfrak{g}_{-\theta}\)  satisfy the relations
\(
(e_{\theta}|e_{-\theta} )=1\) and  \(
\left[ e_{\theta},e_{-\theta}\right ]=h_{\theta}
	\), then we have a Chevalley  generator \(\{\,t\otimes e_{-\theta},\, t^{-1}\otimes e_{\theta},\,-t^0\otimes h_\theta+a\}\) for  the subalgebra \((\mathbf{C}[t,t^{-1}]\otimes\mathfrak{g})\oplus(\mathbf {C}a)\) and adding \(d\) we have the Chevalley generators of the affine Lie algebra, \cite{C}, \cite{K} and \cite{W}.    In the sequel we shall do a similar observation for our Lie algebra \(\widehat{\mathfrak{g}}\).   
  We put
\begin{eqnarray*}
\kappa &=\phi^{+(1,0,1)}\,,\quad 
\kappa_{\ast}&=\,\sqrt{-1}\left(\begin{array}{c} \overline z_2 \\  \overline z_1\end{array}\right)
=-\,\frac{\sqrt{-1}}{\sqrt{2}}\phi^{+(1,1,2)}+\,\frac{\sqrt{-1}}{2}(\phi^{-(0,0,0)}-\phi^{+(1,0,1)})\,.\\[0.5cm]
\mu&=\phi^{-(0,0,0)}\,,\quad 
\mu_{\ast}&=\sqrt{-1}\left(\begin{array}{c} \overline z_2 \\ -\overline z_1\end{array}\right)
=-\,\frac{\sqrt{-1}}{\sqrt{2}}\phi^{+(1,1,2)}-\,\frac{\sqrt{-1}}{2}(\phi^{-(0,0,0)}-\phi^{+(1,0,1)})\,.
\end{eqnarray*}
We recall that \(J=\phi^{+(0,0,0)}=\left(\begin{array}{c} 0 \\ -1\end{array}\right)\).

\begin{lemma}
\begin{enumerate}
\item
\begin{equation}  
\kappa\,\kappa_{\ast}=\kappa_{\ast}\,\kappa=\mu\,\mu_{\ast}=\mu_{\ast}\,\mu=\sqrt{-1}\phi^{+(0,0,1)} .
\end{equation}
\item
\begin{equation}
\,c(\kappa,\kappa_{\ast})=\,c(\mu,\mu_{\ast})=1.
\end{equation}
\end{enumerate}
\end{lemma}

   We consider the following vectors of \(\widehat{\mathfrak{g}}\);
\begin{align}
 \widehat f_{J}&=J\otimes e_{-\theta}\,\in \widehat{\mathfrak{g}}_{0\delta-\theta}\,,\quad &
  \widehat e_{J}&=(-J)\otimes e_{\theta} \,\in \widehat{\mathfrak{g}}_{0\delta+\theta}\,,\\[0.2cm]
   \widehat f_{\kappa}&=\kappa\otimes e_{-\theta}\,\in \widehat{\mathfrak{g}}_{\frac12\delta-\theta}\,,\quad &
  \widehat e_{\kappa}&=\kappa_{\ast} \otimes e_{\theta} \,\in \widehat{\mathfrak{g}}_{\frac12\delta+\theta}\oplus \widehat{\mathfrak{g}}_{-\frac32\delta+\theta}\,,\\[0.2cm]
  \widehat f_{\mu}&=\mu \otimes e_{-\theta}\,\in \widehat{\mathfrak{g}}_{-\frac32\delta-\theta}\,,\quad &
  \widehat e_{\mu}&=\mu_{\ast}\otimes e_{\theta} \,\in \widehat{\mathfrak{g}}_{\frac12\delta+\theta}\oplus \widehat{\mathfrak{g}}_{-\frac32\delta+\theta}\,.
\end{align}
Then we have the generators of \(\widehat{\mathfrak{g}}(a)\)  that are given by the following three tuples:
\begin{eqnarray*}
&& \left(\,\widehat{e}_i,\widehat{f}_i,\widehat{h}_i \right) \quad  i=1,2,\cdots,r,\\[0.2cm]
&&\left( \widehat{e}_{\mu}, \widehat{f}_{\mu},\widehat h_{\theta}\right),\quad 
\left( \widehat{e}_{\kappa}, \widehat{f}_{\kappa},\widehat h_{\theta}\,\right),\quad
 \,\left( \widehat{e}_{J}, \widehat{f}_{J},\widehat h_{\theta}\right)\, \,.
\end{eqnarray*}
These three tuples satisfy the following relations.
\begin{proposition}
\begin{enumerate}
\item
\begin{equation}
\left[\,\widehat e_{\pi}\,,\,\widehat f_i\,\right]_{\widehat{\mathfrak{g}}}=\,
\left[\,\widehat f_{\pi}\,,\,\widehat e_i\,\right]_{\widehat{\mathfrak{g}}} =0\,,\quad\mbox{for } \,1\leq i\leq r ,\,\mbox{ and }\, \pi=J,\,\kappa,\,\mu\,.\label{i}
\end{equation}
\item
\begin{equation}
\left[\,\widehat e_{J}\,,\,\widehat f_J\,\right]_{\widehat{\mathfrak{g}}}=\,\widehat  h_\theta\,,
\end{equation}
\item
\begin{equation}\quad
\left[\,\widehat e_{\mu}\,,\,\widehat f_{\mu}\,\right]_{\widehat{\mathfrak{g}}}=\sqrt{-1}\,\widehat  h_\theta+a ,\quad
\left[\,\widehat e_{\kappa}\,,\,\widehat f_{\kappa}\,\right]_{\widehat{\mathfrak{g}}}=\sqrt{-1}\,\widehat  h_\theta\,+a\,.
\end{equation}
\end{enumerate}
\end{proposition}

\end{document}